\documentclass[12pt,letterpaper]{article}
\usepackage[T1]{fontenc}
\usepackage{lmodern}
\usepackage[utf8]{inputenc}
\usepackage[italian,USenglish]{babel}
\usepackage{geometry}
\usepackage{booktabs}
\usepackage{amsmath}
\usepackage{amsthm} 
\usepackage{amsfonts}
\usepackage{dsfont}
\usepackage{cancel}
\usepackage{ amssymb }
\usepackage{bm}
\usepackage{enumitem}
\usepackage{mathtools}
\usepackage{changepage}
\usepackage{amsthm}
\usepackage{verbatim}
\usepackage{amsmath,amssymb}
\usepackage{graphicx}
\usepackage{emptypage}
\usepackage{newlfont}
\usepackage[all,cmtip]{xy}
\usepackage[bottom]{footmisc}
\usepackage[titletoc,title]{appendix}
\usepackage{chngcntr}
\usepackage{apptools}
\usepackage{listings}
\usepackage{color}
\usepackage{sgame, tikz} 
\usepackage{kpfonts}    
\usepackage{natbib}
\AtAppendix{\counterwithin{lem}{section}}
\usepackage{caption}
\usepackage{subcaption}
\usepackage{float}
\usepackage{thmtools,thm-restate}
\usepackage{epigraph}
\usepackage{xr-hyper}
\usepackage[nodayofweek,level]{datetime}
\usepackage[colorlinks=true,linkcolor=blue, allcolors=blue]{hyperref}
\usepackage{sgamevar}
\usetikzlibrary{external}
\theoremstyle{plain} 
 
\newtheorem{prop}{Proposition}
 
\newtheorem{theorem}{Theorem}
\newtheorem{lemma}{Lemma}
\theoremstyle{definition} 
\newtheorem{ex}{Example}

\newtheorem{thm}{Theorem}
\newtheorem{asm}{Assumption}

\newtheorem{manualasminner}{Assumption}
\newenvironment{manualasm}[1]{
  \renewcommand\themanualasminner{#1}
  \manualasminner
}{\endmanualasminner}
\newtheorem{defn}{Definition}
 
\theoremstyle{remark} 
\newtheorem{rmk}[thm]{Remark}

\DeclareMathOperator{\E}{\mathds{E}}
\renewcommand{\P}{\mathds{P}}

\newcommand{\F}{\mathcal{F}}

\newcommand{\R}{\mathds{R}}

\newcommand{\supp}{\text{supp}}
\renewcommand{\1}{\mathds{1}}
\newcommand\eqas{\stackrel{a.s.}{=}}

\newcommand{\argmax}{\operatornamewithlimits{argmax}}

\newcommand{\X}{\mathcal{X}}

\usepackage{setspace}
\theoremstyle{definition}

\reversemarginpar
\newdate{draftdate}{09}{12}{2021}
\usdate
\usepackage[nameinlink]{cleveref}
\crefname{manualasm}{assumption}{assumptions}
\crefalias{prop}{proposition}
\crefname{claim}{claim}{claims}
\crefname{ex}{example}{examples}
\crefname{defn}{definition}{definitions}
\crefname{rmk}{remark}{remarks}

\begin{document}

\title{Dynamic Pricing with Limited Commitment\thanks{This paper has been previously circulated under the title ``Static Pricing.'' We thank Mohammad Akbarpour, Anirudha Balasubramanian, Gabriel Carroll, Daniel Chen, Laura Doval, Matthew Gentzkow, David Kreps, Giacomo Mantegazza,  Paul Milgrom, Michael Ostrovsky, David Ritzwoller, Ilya Segal, Jesse Shapiro, Andy Skrzypacz, Takuo Sugaya, Bob Wilson, and Weijie Zhong for helpful comments.}}
\author{Martino Banchio\thanks{Graduate School of Business, Stanford University. Email: mbanchio@stanford.edu}\and Frank Yang\thanks{Graduate School of Business, Stanford University. Email: shuny@stanford.edu}}
\date{\displaydate{draftdate}}
\maketitle
\begin{abstract}
A monopolist wants to sell one item per period to a consumer with evolving and persistent private information. The seller sets a price each period depending on the history so far, but cannot commit to future prices. We show that, regardless of the degree of persistence, any equilibrium under a D1-style refinement gives the seller revenue no higher than what she would get from posting all prices in advance.

\vspace{0.1in}
\end{abstract}

\bigskip
\newpage

\section{Introduction}

As electronic commerce and personalized sales become prevalent, dynamic pricing has received a surge of attention (e.g., \citealt{kehoe2020dynamic}).\footnote{The Wall Street Journal reports, ``The Journal identified several companies, including Staples, Discover Financial Services, Rosetta Stone Inc. and Home Depot Inc., that were consistently adjusting prices and displaying different product offers based on a range of characteristics that could be discovered about the user.'' \textit{Websites Vary Prices, Deals Based on Users' Information}, The Wall Street Journal, December 24, 2012.}  Consider a monopolist (she) selling nondurable goods repeatedly to a consumer (he) with evolving private information about his value. The buyer's private values are positively correlated across periods. The seller adjusts prices dynamically based on the information learned through purchase histories, but cannot commit to future prices. Can the seller improve upon simply posting all prices in advance? 

In this paper, we study dynamic pricing with evolving private information and limited commitment. In the benchmark model, we consider a two-period game and show that regardless of the degree of persistence of the private information, any equilibrium satisfying a D1-style refinement gives the seller revenue no higher than what she would get from posting all prices in advance.

This result contrasts sharply with the case of full commitment. A large and growing literature on dynamic mechanism design highlights the value in the ability to dynamically respond to the buyer's evolving private information (\citealt{Baron1984}, \citealt{Courty2000}, \citealt{Eso2007},  \citealt{Pavan2014}). For example, in a repeated-sales model, \citet{Battaglini2005} shows that for any imperfect correlation of the buyer's private values, the seller's average profits converge to the first-best through long-term contracting as the discount factor converges to one.

Instead of analyzing general mechanisms, we restrict attention to (potentially randomized) pricing strategies.\footnote{Among all selling formats, pricing is the most common one in online retail markets; for example, roughly 90\% of the listings in 2015 on eBay are posted price listings \citep{einav2018auctions}.}\textsuperscript{,}\footnote{In a related but different setting (a single durable good with constant values), \cite{Skreta2006} shows that it is without loss of generality to consider pricing strategies among all sequentially rational mechanisms.} By analyzing the equilibria of a dynamic pricing game with evolving private information, we show that the benefits of the flexibility in dynamic pricing hinge critically on the long-term commitment power. 

An important special case arises when the consumer's value distribution is stationary (i.e., identical marginals) but exhibits positive serial correlation. Our result implies that the monopolist cannot do better than simply maintaining the static monopoly price. An extreme case of such a stationary setting is perfect correlation (i.e., constant types). In this extreme case, our result follows from the classic analysis of \citet{Stokey1979}, who shows that the optimal long-term contract with full commitment is static. However, as pointed out earlier, as soon as the correlation is not perfect, the seller can extract more revenue using long-term contracts by conditioning future prices on the buyer's previous decisions.

Why does dynamic pricing, so effective when the seller has long-term commitment power, become futile when the seller has limited commitment? With evolving private information, an important lesson from dynamic mechanism design is that early contracting at the initial stage prevents the buyer from capitalizing on his future information rents. However, such early contracting requires substantial commitment power. Because the types are persistent over time, it is sequentially rational for the seller to post a high price following an acceptance of the first-period offer and a low price following a rejection. As in the constant-type case, the \textit{ratchet effect} appears --- the buyer anticipates the effect of his present choice on future offers: he could give up the purchase today in exchange for a better price tomorrow. As our analysis shows, if the seller could commit to history-contingent prices, then the price schedule would never be time-consistent. 

However, the intuition is incomplete. In environments where more information arrives in the future, the buyer also worries less about his decisions revealing the current private information. The ratchet effect becomes weaker as the persistence of private information decreases. For example, \citet{Kennan2001} studies a version of our bargaining game with binary types. In contrast to the constant-type case of \citet{Hart1988}, he finds that ``when the valuation is persistent, but not fully permanent, the ratchet effect is not so potent: learning does occur in equilibrium.'' In light of his analysis, our result shows that with persistent private information, even though learning is possible in equilibrium, it is unprofitable. The ratchet effect is potent enough to eliminate the potential gains from dynamic pricing. 

The key trade-off of commitment versus learning requires a precise understanding of how equilibrium payoffs depend on the joint distribution of buyer's values. However, as shown in \citet{Kennan2001}, there is a plethora of equilibria even when the buyer's types are binary.  Characterizing all equilibria in games involving the ratchet effect is often intractable (see e.g., \citealt{Laffont1988}).  Our proof strategy connects the dynamic pricing game back to a constrained mechanism design problem. Instead of characterizing all equilibria, we identify a set of necessary conditions on the allocation rule satisfied in all equilibria, and introduce an auxiliary dynamic mechanism design problem constrained by these necessary conditions. The value of this auxiliary problem provides an upper bound on the seller's payoff. We show that posting the monopoly price for each good precisely attains the upper bound. 

In the benchmark model, we assume that the buyer has additive values across periods. In \Cref{sec:comps}, we also show that the same result holds even if the products may be complements. We also provide examples to show that when the values are negatively correlated or when the goods are substitutes, the ratchet effect need not appear, which can then lead to dynamic pricing being profitable even under limited commitment (see \Cref{rmk:negative} and \Cref{rmk:substitutes}). In \Cref{sec:multi-mp}, we introduce a multiple-period version of our game in which the seller always has an option to commit to a price for each unsold good, and show that a similar conclusion holds for arbitrary finite time horizons, under a stronger distributional assumption.

While our result is a theoretical demonstration of the detrimental power of the ratchet effect, it also has practical relevance. First, not all sellers in reality have long-term commitment power.  Even though dynamic pricing policies such as loyalty cards and reward programs are used in practice, attempts to renege, modify, or even cancel them are not rare. In recent years, several class-action law suits have been filed against loyalty programs of multiple companies, including AutoZone (Hughes et al. v. AutoZone Parts et al., 2016), Staples (Torczyner v. Staples, 2016), and AT\&T (Palumbo v. AT\&T, 2021). Second, because of these issues, such long-term pricing requires substantial trust which also depends on the cultural values in specific markets (e.g., \citealt{thompson2015loyalty} suggest that in countries with high uncertainty avoidance culture, consumers tend to find loyalty programs less appealing). Third, machine learning based pricing algorithms, increasingly popular in online market places, often directly incorporate purchase histories as one of the attributes (e.g., \citealt{gupta2014machine}). In all these scenarios, our result suggests that the seller may be better off by posting the price for each good in advance.\footnote{We also show that even if the seller is allowed to commit to history-dependent prices, as long as she cannot use policies like repeated-purchase discounts by setting the price after acceptance lower than the price after rejection, the optimal mechanism is still to simply post all prices in advance (see \Cref{rmk:commit}).} Of course, maintaining these prices also requires some commitment power. But arguably such a strategy is more convenient and has fewer trust issues, especially in the important case of stationary value distributions where the seller only needs to maintain a single price (e.g., \citealt{einav2018auctions}).\footnote{\citet{einav2018auctions} document the decline of auctions used in online marketplaces and argue that the convenience of posted prices is one of the main drivers. Posting all prices in advance is also known as price commitment, which has been argued to be a common pricing strategy (see e.g., \citealt{Jing2011}). }

Our model builds on a large literature on behavior-based price discrimination (e.g., \citealt{Hart1988}, \citealt{Villas-Boas2004}, \citealt{Acquisti2005}), which is surveyed by \citet{Fudenberg2006, Fudenberg2012} and \citet{Acquisti2016}. Most work in this literature assumes that the consumer's values are perfectly correlated across periods.\footnote{Two notable exceptions are \citet{taylor2004consumer} and \citet{Kennan2001}. Both papers study a binary-type setting and allow for positive correlation of types. \citet{taylor2004consumer} characterizes the pure-strategy equilibria; \citet{Kennan2001} shows that there is often plethora of mixed-strategy equilibria and characterizes a particular class of them. In contrast, our model allows for a continuum of types in each period, and our result holds for all mixed-strategy equilibria under a D1-style refinement. } This assumption rules out both the possibility of differentiated goods and the possibility of taste shocks. More importantly, as indicated earlier, the case of perfect correlation is a knife-edge case in which the effectiveness of  behavior-based price discrimination is muted even with full commitment. Our result shows that the detrimental power of the ratchet effect, however, is not knife-edge --- it eliminates the gains from dynamic pricing for any degree of persistence even when the optimal contract itself is dynamic. 

The study of the ratchet effect goes back to the classic work of \citet{freixas1985planning}, \citet{Laffont1988}, and \citet{Hart1988}. Recently, there has been a renewed interest in better understanding its impact (e.g., \citealt{Gerardi2020},   \citealt{Bonatti2020}). In particular, \citet{Bonatti2020} study how strategic consumers interacting with a sequence of firms (in continuous time) reduce their demand to manipulate their scores that aggregate past purchase histories. In contrast, in our model, the monopolist is long-lived and can condition her price offer on the entire sequence of purchase histories. 

Our result connects the ratchet effect to the growing literature of dynamic mechanism design, which is surveyed by \citet{Pavan2017} and \citet{Bergemann2019}. As \citet{Pavan2017} points out, ``[The limited commitment] literature assumes information is static, thus abstracting from the questions at the heart of the dynamic mechanism design literature. I expect interesting new developments to come out from combining the two literatures.'' Indeed, as our result shows, introducing limited commitment can give a qualitatively different answer to important economic questions such as whether dynamic pricing is profitable. While we restrict attention to randomized pricing mechanisms (instead of general mechanisms), this approach enables us to (i) provide an empirically relevant benchmark result and (ii) connect the economics of the ratchet effect, which is at the heart of limited commitment models, to dynamic mechanism design.\footnote{Our approach is conceptually similar to \citet{Liu2019} who study limited commitment while restricting attention to auctions. In contrast to us, they study a durable good monopoly setting in which the game ends after the good is traded, and do not introduce evolving private information.}\textsuperscript{,}\footnote{A different form of limited commitment is considered by \citet{deb2015dynamic} who focus on the arrival of new consumers instead of the evolution of private information. Also, see \citet{Doval2020} for general mechanism design with limited commitment. Our approach has the additional benefit of being tractable. Unlike \citet{Strulovici2017}, \citet{Doval2019}, and \citet{Gerardi2020}, our model allows for a continuum of types in each period.} We also show that with evolving private information, equilibrium refinement plays a crucial role in modeling limited commitment (see \Cref{rmk:d1-example}), and introduce a notion of D1 criterion in the spirit of \citet{Banks1987} for our game (see \Cref{app:d1}).

Our paper also relates to a more applied literature on dynamic pricing and price commitment, especially for experience goods for which the consumer learns his value after consuming the first unit (e.g., \citealt{Cremer1984} and \citealt{Jing2011}). The key difference is that in such models, the consumer is often ex ante uninformed, and the purchase history only reveals whether the consumer becomes informed about his value. Thus, the ratchet effect is absent. Our model allows for consumption-independent learning and focuses on the ratchet effect induced by the positive correlation of values. Of course, in doing so, we abstract away many other important features of repeated sales (e.g., finite inventories, stochastic arrivals, and search costs). There is a large literature on dynamic pricing in management science that focuses on these features (e.g., \citealt{elmaghraby2003dynamic}, \citealt{su2007intertemporal}, \citealt{aviv2008optimal}, \citealt{cachon2015price}).

The rest of the paper proceeds as follows. \Cref{sec:model} presents the model. \Cref{sec:main} presents the main result and its proof. \Cref{sec:ext} shows the robustness of the result to two variations: complementary products (\Cref{sec:comps}) and multiple periods (\Cref{sec:multi-mp}). \Cref{sec:disc} concludes.

\section{Model}\label{sec:model}

A seller wants to sell a nondurable good in each period $t \in \{1, 2\}$ to a buyer, with a cost normalized to $0$. The buyer privately observes his value $\theta_t\in[\underline{\theta}_t,\overline{\theta}_t] \subset \mathbb{R}_+$ at the beginning of each period. Let $F_1(\theta_1), \, F_2(\theta_2)$ denote the marginal distributions of $\theta_1, \, \theta_2$ respectively, and $F_2(\theta_2|\theta_1)$ denote the distribution of $\theta_2$ conditional on $\theta_1$. The distributions are commonly known. We assume that $F_1(\theta_1)$ is absolutely continuous with a density denoted by $f_1(\theta_1)$, and $F_{2}(\theta_{2}|\theta_1)$ is continuously differentiable in $(\theta_1, \theta_{2})$ with a conditional density denoted by $f_{2}(\theta_{2}|\theta_1)$. We also assume that $f_1(\theta_1) > 0$ for all $\theta_1 \in [\underline{\theta}_1, \overline{\theta}_1]$.

At the beginning of each period $t$, the seller makes a take-it-or-leave-it offer $p_t$ to the buyer for the product in period $t$. The buyer decides whether to accept the offer. After the buyer makes his decision, the game moves to the next period. Let $x_t \in \{0,1\}$ denote the decision of the buyer in period $t$. We assume that both the seller and buyer have quasilinear preferences that are additively separable across the two periods, with a common discount factor $\delta \in (0, 1]$. So the seller's period-$t$ payoff is given by $u_t^{S}(x_t,p_t) = \delta^{t-1} x_t p_t $, and the buyer's period-$t$ payoff is given by  $u_t^B(\theta_t,x_t,p_t) = \delta^{t-1} x_t (\theta_t - p_t)$. 

\subsubsection*{Strategies and Beliefs}

Let $h_t$ and $\hat{h}_t$ denote the public history observed by the seller and the history observed by the buyer up to period $t$, respectively. The seller's strategy $\sigma$ maps a public history $h_t$ to a distribution on $\R$ representing the price announced in period $t$. A buyer's strategy $\tau$ maps a history $\hat{h}_t$ to a distribution on $\{0, 1\}$ representing his purchase decision in period $t$. The seller forms belief $\mu(h_t)$ about $\theta := (\theta_1, \theta_2)$ given a public history $h_t$. 

\subsubsection*{Equilibrium}
Our choice of equilibrium concept is perfect Bayesian equilibrium with a further restriction on the seller's off-path beliefs. A perfect Bayesian equilibrium satisfying the D1 criterion (PBE-D) of this game is a tuple $(\sigma, \tau, \mu)$ such that the following hold:
\begin{itemize}
    \item Given $\mu(h_t)$ and $\tau$, $\sigma$ is sequentially rational at every public history $h_t$.
    \item Given $\sigma$, $\tau$ is sequentially rational at every history $\hat{h}_t$.
    \item $\mu(h_t)$ is derived via Bayes' rule whenever possible. 
    \item Seller's off-path beliefs satisfy a D1 criterion, formally defined in \Cref{app:d1}. 
\end{itemize}

The only difference from a standard PBE is in the last point. We require that an equilibrium satisfy a D1 criterion in the spirit of \citet{Banks1987}, at any off-path history. The seller cannot update her belief via Bayes' rule at some history $h_t$ when all types of the buyer accept (or reject) in equilibrium, but she observes a rejection (or acceptance). Intuitively,  the refinement requires that when observing an unexpected deviation, the seller believes that the buyer is of the type that could benefit the most from the deviation. Even though our model is not a signaling game, \citet{Banks1987}'s notion can be extended to our setting, as formally shown in \Cref{app:d1}.

\subsubsection*{Distributional Assumptions}

We impose the following assumptions for our main result. 

\begin{asm}[MLRP]\label[manualasm]{asm:mlrp}
For any $\theta_1' > \theta_1$, the conditional likelihood ratio $\frac{f_2(\theta_2 | \theta_1')}{f_2(\theta_2 | \theta_1)}$
is strictly increasing in $\theta_2$.
\end{asm}

\Cref{asm:mlrp} is equivalent to strict affiliation of $(\theta_1, \theta_2)$ as in \citet{Milgrom1982}. For example, consider an AR(1) process of the form $\theta_2 = \alpha \theta_1 + \epsilon$, where $\epsilon$ is drawn from a Gaussian distribution. Then \Cref{asm:mlrp} is equivalent to $\alpha > 0$.

\begin{asm}[$1/\delta$-Lipschitz]\label[manualasm]{asm:Lipschitz}
For any $\theta'_1 > \theta_1$, 
\[\E[\theta_2|\theta_1']- \E[\theta_2|\theta_1] < \frac{1}{\delta} (\theta_1'-\theta_1) \,.\]
\end{asm}
\Cref{asm:Lipschitz} says that a change in the first-period type leads to a comparable change in the second-period type in expectation. For any AR(1) process $\theta_2 = \alpha \theta_1 + \epsilon$, \Cref{asm:Lipschitz} is equivalent to $\alpha < \frac{1}{\delta}$. 

To state our last assumption, we introduce the notion of an impulse response function (\citealt{Pavan2014}):
\begin{equation*}\label{eq:impulse}
    I(\theta_1,\theta_2) :=  -\frac{\partial F_2(\theta_2|\theta_{1})/\partial \theta_{1}}{f_2(\theta_2|\theta_{1})}\,.
\end{equation*}
\noindent Intuitively, the impulse response function measures the effect of a small change in $\theta_1$ on $\theta_2$. For any AR(1) process, $\theta_2 = \alpha \theta_1 + \epsilon$, we have $I(\theta_1, \theta_2) = \alpha $.

\begin{asm}[Regularity]\label[manualasm]{asm:vvfunction} The second-period virtual value function
\[\psi(\theta_1, \theta_2) := \theta_2 - \frac{1 - F_1(\theta_1)}{f_1(\theta_1)} I(\theta_1, \theta_2) \]
is non-decreasing in both $\theta_1$ and $\theta_2$.
\end{asm} 
\Cref{asm:vvfunction} is the dynamic version of Myerson's regularity condition and commonly 
imposed in dynamic mechanism design for tractability (see e.g., \citealt{Bergemann2019}). For any AR(1) process, we have $\psi(\theta_1, \theta_2) = \theta_2 - \frac{1 - F_1(\theta_1)}{f_1(\theta_1)}\alpha$, and hence \Cref{asm:vvfunction} holds if the distribution $F_1$ has monotone hazard rate.

\begin{rmk}\label{rmk:d1}
Under the distributional assumptions, we show in \Cref{app:d1} that the D1 criterion is satisfied if the seller believes the buyer is of type $\underline{\theta}_1$ when observing an off-path rejection, and of type $\overline{\theta}_1$ when observing an off-path acceptance. Moreover, we also show  in \Cref{app:d1}  that for any PBE that satisfies the D1 criterion, there exists an outcome-equivalent PBE with off-path beliefs of this form.\footnote{We say two equilibria
are outcome-equivalent if they almost surely induce the same allocation and the same total discounted payment.} Therefore, for all of our results, it is without loss of generality to focus on the PBE with such off-path beliefs. See \Cref{rmk:d1-example} for the importance of restricting off-path beliefs in modeling limited commitment. 
\end{rmk}

\section{Main Result}\label{sec:main}

Let $\pi^*_t = \max_{p} p(1 - F_t(p))$ be the standard monopoly revenue with respect to the marginal distribution in period $t$. Our main result says that any equilibrium gives the seller revenue no higher than what she would get by simply posting all prices in advance:
\begin{restatable}{theorem}{thm:main}\label{thm:main} Under \Cref{asm:mlrp,asm:Lipschitz,asm:vvfunction}, 
\begin{itemize}
    \item[(i)] there exists a PBE-D; 
    \item[(ii)] in every PBE-D, the seller's revenue is no higher than $\sum_t \delta^{t-1} \pi^*_t$. 
\end{itemize}
\end{restatable}

We leave the proof of part (i) of \Cref{thm:main} to \Cref{app:existence}. The proof proceeds by considering an auxiliary one-shot continuous game, showing the existence of an equilibrium in the auxiliary game using the result of \citet{Glicksberg1952}, and constructing a PBE-D using the continuity properties of the equilibria of the auxiliary game.

We prove part (ii) of \Cref{thm:main} in \Cref{subsec:proof}. The intuition behind the proof of this revenue bound can be understood as follows. Suppose hypothetically that the seller could commit to any history-contingent price path that specifies the first-period price $p_1$, the second-period price after an acceptance $p_A$, and the second-period price after a rejection $p_R$. We show that the seller would always commit to some $p_A < p_R$.  Before explaining this claim, we first explain why this implies that the seller's revenue in equilibrium is no higher than the sum of discounted monopoly revenues.

When the seller observes an acceptance, she believes that the buyer has a high type. Because the types are persistent, sequential rationality implies that the seller always sets $p_A \geq p_R$ in equilibrium. But by our earlier claim, the seller would prefer to set $p_A < p_R$. Thus, the seller with limited commitment cannot obtain a revenue higher than what she would get by letting $p_A = p_R$, or equivalently by simply posting all prices in advance. \Cref{fig:intuition} illustrates this argument with stationary distributions. When types are perfectly correlated, the optimal commitment prices $p_A, \, p_R$ are both equal to the static monopoly price $p^*$. When the correlation becomes imperfect, the seller benefits from committing to a lower $p_A$ and a higher $p_R$, but sequential rationality constrains the seller in exactly the \textit{opposite} direction.

\begin{figure}[h]
	\centering
	\includegraphics{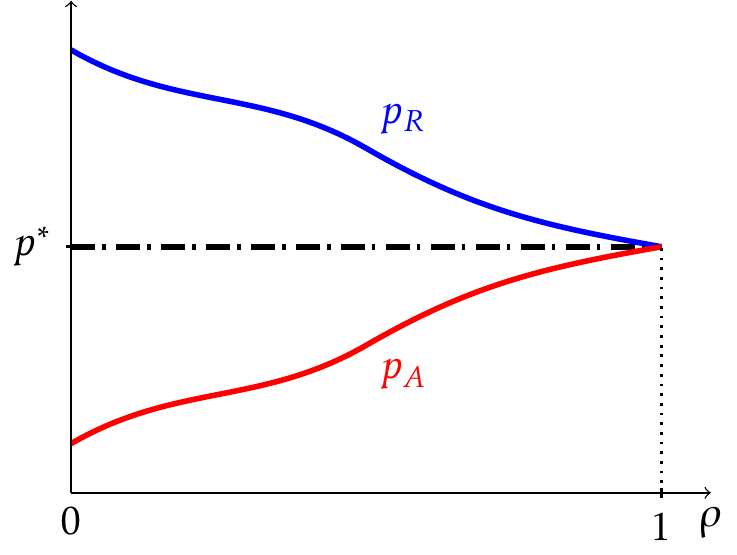}
    \caption{Illustration of the seller's optimal commitment prices}
    \label{fig:intuition}
\end{figure}

To see why the seller would prefer to set $p_A < p_R$ if allowed to commit, it is helpful to consider the following simple example (assuming $\delta = 1$): 

\begin{ex}\label{ex:ind}
The buyer has independent values across the two periods, and in each period $\theta_t$ is drawn uniformly from the interval $[1, 2]$. 
\end{ex}

Note that the monopoly price for the uniform distribution on $[1, 2]$ is $1$. Thus, if the seller posts prices in advance, then she will post price $1$ in both periods and collect a total revenue of $2$. Now, suppose the seller commits to $(p_1  = 1.5, p_A = 1, p_R = 2)$. Consider how the buyer of type $\theta_1$ decides in the first period. If he rejects the first period offer, then he gets $0$ payoff because the second-period offer is too high. If he accepts, then his expected payoff is given by 
\[\theta_1 - p_1 + \E[(\theta_2 - p_A)_+ |\theta_1]= \theta_1 - p_1 + \E[\theta_2 - p_A |\theta_1] = \theta_1 - 1.5 + \E[\theta_2] - 1  =  \theta_1 - 1 \geq 0\]
because $\theta_2$ is always no less than $p_A = 1$ and $\theta_2$ is independent of $\theta_1$. Therefore, all types of buyer will accept the first-period offer and then accept the second-period offer. The seller gets a total revenue of $1.5  + 1  = 2.5 > 2$. 

The first-period price $p_1 = 1.5$ includes both the monopoly price for the first good and an option price of $0.5$ for getting the second good at $p_A = 1$. By committing to a low $p_A$ and a high $p_R$, the seller can add a surcharge on top of the monopoly price to sell the second good ``in expectation.''  The buyer only gets information rents from his private information about $\theta_1$, as the second-period surplus $\E[\theta_2] = 1.5$ is completely extracted. 

This observation does not hinge on the assumption that the types are independent across periods. As long as there is some uncertainty in the second-period types, the seller in general benefits from committing to some $p_A < p_R$. This effect echoes the power of ``early contracting'' in the literature on dynamic mechanism design: contracting before the agent's future private information realizes allows the designer to extract more surplus. In the proof, we use the techniques from dynamic mechanism design to make a precise statement about when the seller prefers to commit to $p_A < p_R$. 

\begin{rmk}\label{rmk:d1-example}
The sequential rationality constraint only comes into play if the seller is at an on-path history. Without any equilibrium refinement, the seller can have implausible beliefs at some off-path history in equilibrium. Such beliefs can loosen the sequential rationality constraint and effectively grant the seller more and sometimes full commitment power. To see this concretely, consider the following example: 
\end{rmk}

\begin{ex}\label{ex:d1}
Type $\theta_1$ is drawn uniformly from $[1,2]$. If type $\theta_1 =2$, then type $\theta_2=2$; otherwise, $\theta_2$ is drawn uniformly from $[1,2]$ independently of $\theta_1$.
\end{ex}

\Cref{ex:d1} differs from \Cref{ex:ind} only when $\theta_1 = 2$, but the seller can now obtain the full-commitment outcome that we have just described ($p_1 = 1.5, p_A = 1, p_R = 2$) in a perfect Bayesian equilibrium if there is no restriction on the off-path beliefs.\footnote{This outcome is actually optimal even if the seller could commit to any mechanism. To see this, note that the seller can never do better than observing the realization of $\theta_2$ and then committing to a direct mechanism mapping a type report $\hat{\theta}_1$ to an allocation and a transfer. In that case, the optimal mechanism is to allocate both goods and charge a price of $1 + \theta_2$, which generates a revenue of $2.5$.} All types solve the same problem in the first period as before, except the type $\theta_1 = 2$, but this type also prefers accepting the first-period offer. So all types accept in both periods. Now suppose the seller holds a belief $\mu$ concentrated on $\theta_1=2$ after a first-period rejection (which is an off-path history). This belief makes the price $p_R = 2$ sequentially rational, effectively granting the seller full commitment power. The D1 refinement eliminates such implausible beliefs by essentially requiring the seller to believe that the buyer who deviates must be of the type most likely to profit from the off-path deviation. As noted in \Cref{rmk:d1}, the seller would believe that the buyer is of type $\theta_1=1$ after observing an off-path rejection, resulting in $p_A \geq p_R$ even at an off-path history.

\subsection{Proof of the Main Result}\label{subsec:proof}

The proof of the revenue bound, part (ii) of \Cref{thm:main}, proceeds in three steps: (1) identify a set of necessary conditions that are satisfied in every PBE-D, (2) consider a mechanism design problem constrained by this set of necessary conditions, and (3) show that the optimal solution to this relaxed problem can be implemented by simply posting all prices in advance. For expositional convenience, in this section we prove the statement for all pure-strategy PBE-D. In \Cref{app:main}, we extend the proof to cover all PBE-D. 

\paragraph{Step 1:} Let $p_A,\, p_R$ be the seller's equilibrium choice of prices posted after observing an acceptance and a rejection, respectively. Let $x_1(\theta_1)$ be the equilibrium first-period allocation. We show that a key restriction imposed by sequential rationality is the following set of monotonicity constraints: 

\begin{restatable}{prop}{pricemon}\label[proposition]{prop:pricemon}
In every pure-strategy PBE-D, $x_1$ is a threshold rule and $p_A \geq p_R$.
\end{restatable}

Before proving this claim, we collect two lemmas.

\begin{restatable}{lemma}{MLRPone}\label{lem:MLRPone}
For any $k \in (\underline{\theta}_1, \overline{\theta}_1)$,
\[\frac{f(\theta_2 | \theta_1 \geq k) }{f(\theta_2 | \theta_1 < k) }\]
is strictly increasing in $\theta_2$, i.e., $F_2(\ \cdot \ | \theta_1 \geq k) \succ_{L} F_2( \, \cdot\, | \theta_1 < k)$.\footnote{We use the notation $\succ_{L}$ for strict likelihood-ratio dominance, i.e., the likelihood ratio of two distributions is a strictly increasing function.}
\end{restatable}
\begin{proof}
See \Cref{app:lem1}. 
\end{proof}

This lemma is a direct consequence of \Cref{asm:mlrp}. It sorts the seller's posterior belief in the second period, assuming a threshold rule in the first period. 

\begin{restatable}{lemma}{MLRPtwo}\label{lem:MLRPtwo}
If $F_2 \succ_L F_1$ and $p_i$ is an optimal monopoly price under $F_i$, then $p_2 \geq p_1$.
\end{restatable}
\begin{proof}
See \Cref{app:lem2}. 
\end{proof}

Because likelihood ratio ordering implies inverse hazard rate ordering, the above observation follows immediately from the monotone selection theorem of \citet{Topkis1998}.

\begin{proof}[Proof of \Cref{prop:pricemon}]
In the second period, the buyer accepts when the price is below his second-period type $\theta_2$. Thus, in the first period, type $\theta_1$ accepts a price $p$ only if
\begin{equation}\label{eq:pricemon}
(\theta_1 - p) + \delta \E[(\theta_2 - p_A)_{+}|\theta_1] \geq \delta  \E[(\theta_2 - p_R)_{+}|\theta_1]    \,.
\end{equation}
Let $m(\theta_2) := (\theta_2 - p_R)_{+}- (\theta_2 - p_A)_{+}$, $h(\theta_1) := \E[m(\theta_2) | \theta_1]$ and $g(\theta_1) := \frac{1}{\delta}(\theta_1-p)$. So type $\theta_1$ accepts only if $g(\theta_1) \geq h(\theta_1)$ and always accepts if the inequality is strict. Note that $m(\theta_2)$ is $1$-Lipschitz and has a.e. derivative $m'\leq 1$. For any $\theta_1' > \theta_1$, integration by parts yields
\begin{align}\label{eq:h}
 h(\theta'_1) - h(\theta_1) = \E[m(\theta_2) | \theta_1'] - \E[m(\theta_2) | \theta_1] 
&= \int_0^\infty [\P(\theta_2 > s|\theta'_1) -  \P(\theta_2 > s|\theta_1)]m'(s) ds \\  
& \leq \E[\theta_2|\theta_1'] - \E[\theta_2| \theta_1]  < \frac{1}{\delta}(\theta_1' - \theta_1) = g(\theta_1') - g(\theta_1) \nonumber
\end{align}
where we have also used $\theta_2 | \theta_1'$ $\succeq_{FOSD}$ $\theta_2 | \theta_1$ (implied by \Cref{asm:mlrp}) and the Lipschitz condition on $\E[\theta_2|\ \cdot \ ]$ (\Cref{asm:Lipschitz}). Hence, $g$ crosses $h$ at most once from below. The first-period allocation is characterized by a cutoff type $k$, proving the first part. 

Suppose the cutoff type $k \in  (\underline{\theta}_1, \overline{\theta}_1)$. By \Cref{lem:MLRPone}, $F_2(\ \cdot \ | \theta_1 \geq k) \succ_L F_2( \, \cdot\, | \theta_1 < k)$. Then by \Cref{lem:MLRPtwo}, $p_A \geq p_R$, because $p_A$ must be optimal under belief $F_2( \, \cdot\, | \theta_1 \geq k)$ and $p_R$ must be optimal under belief $F_2( \, \cdot\, | \theta_1 < k)$. If $k = \underline{\theta}_1$, then a rejection is off-path and it is without loss of generality to let the seller form a belief $F_2( \ \cdot \ |\theta_1 = \underline{\theta}_1)$  as noted in \Cref{rmk:d1}. Hence, $p_A \geq p_R$ as $F_2(\, \cdot\, |\theta_1 \geq \underline{\theta}_1) \succ_L F_2(\, \cdot\, |\theta_1 = \underline{\theta}_1)$. The same argument holds for $k = \overline{\theta}_1$.  
\end{proof}

\paragraph{Step 2:} Because the two monotonicity constraints must hold in any equilibrium, if the seller can commit to any price mechanism subject to these constraints, then she can always replicate any (pure-strategy) PBE-D. Therefore, providing her the commitment power subject to these constraints forms a relaxed problem. 

Specifically, the relaxation is done as follows:

\begin{minipage}[t]{0.43\linewidth}
\vspace{0.1cm}
\textbf{Original}:\\
Select a seller-optimal PBE-D 
\begin{itemize}
\item[1.] Buyer's IC constraints
\item[2.] Seller's sequential rationality
\item[3.] Belief consistency\\
\end{itemize}
\end{minipage}
\hfill
\begin{minipage}[t]{0.7\linewidth}
\vspace{0.1cm}
\textbf{Relaxed}:\\ 
Design a revenue-maximizing mechanism 
\begin{itemize}
\item[1.] Buyer's IC constraints
\item[2.] $x_1(\theta_1) = \1_{\theta_1 \geq k}$
\item[3.] $x_2(\theta_1, \theta_2) = x_1(\theta_1) \1_{\theta_2 \geq p_A} + (1 - x_1(\theta_1))  \1_{\theta_2 \geq p_R}$, 

$p_A \geq p_R$\\
\end{itemize}
\end{minipage}

In this relaxation, we replace the sequential rationality and belief consistency constraints with the constraints imposed directly on the allocation rules. The space of price mechanisms also constrains the set of feasible allocation and payment rules jointly. Let $\X$ denote the set of all incentive-compatible dynamic direct mechanisms satisfying the feasibility and monotonicity constraints. The seller then solves the following problem: 
\begin{equation}\label{eq:objective}
\max_{(x, t) \in \X} \E\bigg[\theta_1 x_1(\theta_1) + \delta \theta_2 x_2(\theta_1, \theta_2) - V_{x, t}(\theta_1) \bigg]
\end{equation}
where $V_{x, t}(\theta_1)$ is the expected payoff of a type $\theta_1$ buyer under the mechanism $(x, t)$. (We use $t$ to denote both the transfer and time. Its meaning should be clear from the context.) As in static mechanism design, local incentive compatibility constraints imply a version of payoff equivalence in dynamic environments. In particular, we use the following result from \citet{Pavan2014}:
\begin{prop}[\citealt{Pavan2014}]\label{prop:pe}
For any incentive-compatible direct mechanism $(x, t)$, $V_{x, t}(\theta_1)$ is Lipschitz continuous with the derivative given almost everywhere by 
\begin{equation}\label{eq:ICFOC}
    V'_{x, t}(\theta_1) = \E\Big[x_1(\theta_1) + \delta I(\theta_1, \theta_2) x_2(\theta_1, \theta_2) \, \Big| \, \theta_1 \Big] \,.
\end{equation}
\end{prop}
\begin{proof}
See Theorem 1 in \citet{Pavan2014}.
\end{proof}

With (\ref{eq:ICFOC}) and the usual integration by parts, we can  rewrite objective (\ref{eq:objective}) as
\[\max_{(x, t) \in \X} \E[\theta_1 x_1(\theta_1) + \delta \theta_2 x_2(\theta_1, \theta_2)] - \E\Bigg[\frac{1 - F_1(\theta_1)}{f_1(\theta_1)} \big(x_1(\theta_1) + \delta I(\theta_1, \theta_2)x_2(\theta_1, \theta_2)\big) \Bigg] - V_{x,t}(\underline{\theta}_1)\,.
\]
Importantly, with a price mechanism, the seller cannot always extract all the expected surplus of the lowest type $\underline{\theta}_1$. In particular, the lowest type $\underline{\theta}_1$ can simply reject the first-period offer and obtain a price $p_R$ for the second-period item, which implies that
\[V_{x, t}(\underline{\theta}_1) \geq  \delta \E[(\theta_2 - p_R)_+ | \theta_1 = \underline{\theta}_1]\,.\]
With these observations, we reach the final form of our relaxation 
\begin{equation}\label{eq:relax}
\max_{(x, t) \in \X}  \E[\varphi(\theta_1)x_1(\theta_1)] + \delta\E[ \psi(\theta_1, \theta_2)  x_2(\theta_1, \theta_2)] -  \delta\E[(\theta_2 - p_R)_+ | \theta_1 = \underline{\theta}_1]
\end{equation}
where \[\varphi(\theta_1) := \theta_1 - \frac{1 - F_1(\theta_1)}{f_1(\theta_1)}, \qquad \psi(\theta_1, \theta_2) := \theta_2 - \frac{1 - F_1(\theta_1)}{f_1(\theta_1)} I(\theta_1, \theta_2)\] are the virtual value functions for the first and second period, respectively. Note that if a mechanism $(x^*, t^*)$ solves ($\ref{eq:relax}$) and satisfies $V_{x^*, t^*}(\underline{\theta}_1) = \delta \E[(\theta_2 - p^*_R)_+ | \theta_1 = \underline{\theta}_1]$, then the mechanism also solves ($\ref{eq:objective}$) and the value of ($\ref{eq:relax}$) is exactly its expected revenue. 
\paragraph{Step 3:}  To solve (\ref{eq:relax}), note that under the constraints, the relevant parameters for the allocation rule are just $(k, p_A, p_R)$ with $p_A \geq p_R$. We optimize directly over these: 
\[
\max_{k, p_A, p_R: p_A \geq p_R} \E[\varphi(\theta_1)\1_{\theta_1 \geq k}] + \delta \E[\psi(\theta_1, \theta_2) (\1_{\theta_1 \geq k} \1_{\theta_2 \geq p_A}+ \1_{\theta_1 < k} \1_{\theta_2 \geq p_R}) ] - \delta \E[(\theta_2 - p_R)_+ | \theta_1 = \underline{\theta}_1]\,.\]

Fix any $k$ and consider the subproblem:
\begin{equation}\label{eq:star}
    \max_{p_A, p_R: p_A \geq p_R} \E[\psi(\theta_1, \theta_2) (\1_{\theta_1 \geq k} \1_{\theta_2 \geq p_A}+ \1_{\theta_1 < k} \1_{\theta_2 \geq p_R}) ] - \E[(\theta_2 - p_R)_+ | \theta_1 = \underline{\theta}_1] \,.
\end{equation}
Define $D(\theta_1) := \inf_{\theta_2}\{\psi(\theta_1, \theta_2) \geq 0\}$.\footnote{Moreover, let $D(\theta_1) := \overline{\theta}_2$ if $\psi(\theta_1, \theta_2) < 0$ for all $\theta_2$ on which $\psi(\theta_1, \,\cdot\,)$ is defined; let $D(\theta_1) := \underline{\theta}_2$ if $\psi(\theta_1, \theta_2) \geq 0$ for all $\theta_2$ on which $\psi(\theta_1, \,\cdot\,)$ is defined.} If we allow for any allocation rule and the objective only consists of the first term, then we would set $x_2(\theta_1, \theta_2) = 1$ whenever $\psi(\theta_1, \theta_2) \geq 0$. That is, the second item is only allocated to types that fall above the boundary curve $D$ in the type space $[\underline{\theta}_1, \overline{\theta}_1]\times [\underline{\theta}_2, \overline{\theta}_2]$.  However, this is not a feasible solution to the constrained problem (\ref{eq:star}). To proceed, we make the following key observation:

\begin{restatable}{claim}{key}\label[claim]{claim:key}
For any $k$, there exists an optimal solution to (\ref{eq:star}) with $p_A = p_R$. 
\end{restatable}
\begin{proof}[Proof of \Cref{claim:key}]
Let $z = D(k)$ be the point where $D$ crosses vertical line at $k$. Note that \Cref{asm:vvfunction} implies that $D(\,\cdot\,)$ is a non-increasing function. The existence of an optimal solution follows from standard compactness arguments.  Fix any optimal solution $(p_A, p_R)$ to program (\ref{eq:star}). If $p_A = p_R$, then we are done. Otherwise, there are three cases. 
 
Case (i): $p_R < z < p_A$. As \Cref{fig:proof} shows, by lowering $p_A$ the seller captures only types with $\psi(\theta_1, \theta_2) \geq 0$.  Similarly, by raising $p_R$, the seller eliminates types with $\psi(\theta_1, \theta_2) < 0$. Importantly, the second term in the objective, $ - \E[(\theta_2 - p_R)_+ | \theta_1 = \underline{\theta}_1]$, is also non-decreasing in $p_R$. Therefore, $(z, z)$ must be weakly better than $(p_A, p_R)$, and hence $(z, z)$ is an optimal solution. 

\begin{figure}[h]
    \centering
	\includegraphics{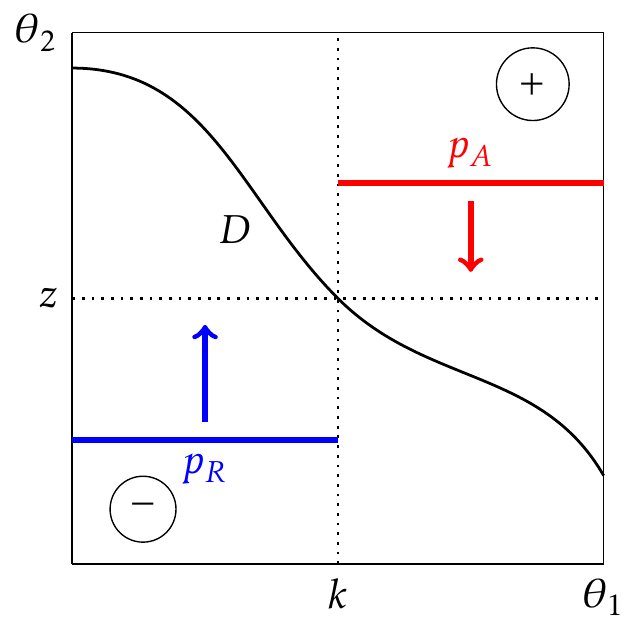}
    \caption{Allocation rule in the type space}
    \label{fig:proof}
\end{figure}

Case (ii): $z \leq p_R < p_A$. By the same reasoning, lowering $p_A$ to $p_R$ weakly increases the objective. Case (iii): $ p_R < p_A \leq z$. By the same reasoning, raising $p_R$ to $p_A$ weakly increases the objective. Thus, there exists an optimal solution with $p_A=p_R$.\end{proof}

Because \Cref{claim:key} holds for any $k$, we can further reduce ($\ref{eq:relax}$) to 
\begin{equation}\label{eq:last}
\max_{k, p_2} \E[\varphi(\theta_1)\1_{\theta_1 \geq k}] + \delta \E[\psi(\theta_1, \theta_2) \1_{\theta_2 \geq p_2} ] - \delta \E[(\theta_2 - p_2)_+ | \theta_1 = \underline{\theta}_1] \,.
\end{equation}
Let $(k^*, p_2^*)$ be an optimal solution to (\ref{eq:last}), which exists by compactness arguments. Consider the mechanism that posts prices $k^*, p_2^*$ for the first-period and second-period items respectively. It implements the allocation rule $x_1(\theta_1) = \1_{\theta_1 \geq k^*}, x_2(\theta_1, \theta_2) = \1_{\theta_2 \geq p_2^*}$ that solves ($\ref{eq:last}$) and hence solves ($\ref{eq:relax}$). In addition, the lowest type gets an expected payoff $V(\underline{\theta}_1) = \delta \E[(\theta_2 - p^*_2)_+ | \theta_1 = \underline{\theta}_1]$.
Therefore, it solves ($\ref{eq:objective}$) and its expected revenue is an upper bound on what the seller can achieve in any pure-strategy PBE-D. Of course, its expected revenue is no higher than $\sum_t \delta^{t-1} \pi^*_t$, which completes the proof. \hfill \textit{Q.E.D.}

\begin{rmk}\label{rmk:mixed-strategy}
In \Cref{app:main}, we extend this proof to cover all PBE-D (including the ones in mixed strategies). In Step 1, we show that the first-period allocation is still a threshold rule and that the distribution of $p_A$ must dominate the distribution of $p_R$ in the sense that any realized $p_A$ is weakly higher than any realized $p_R$. In Step 2, we consider a similar constrained mechanism design problem with the difference that the seller now chooses two distributions subject to a disjoint support constraint. In Step 3, we use \Cref{claim:key} to show that the optimal value of the relaxed problem is again attained by simply posting two deterministic prices, one for each period.  
\end{rmk}

\begin{rmk}\label{rmk:commit}
Note that, as the proof shows, our result also holds in the setting where the seller can commit to any dynamic price mechanism $(p_1, p_A, p_R)$ subject to the constraint $p_A \geq p_R$. Therefore, as indicated earlier, the only way that dynamic pricing can be profitable in our setting is to introduce repeated purchase discounts by setting $p_A < p_R$. Whenever offering repeated purchase discounts is not a feasible strategy (as in the examples discussed in the introduction), our proof shows that the seller cannot do better than simply posting all prices in advance. 
\end{rmk}

\begin{rmk}\label{rmk:refinement}
Moreover, as the proof shows, our result also holds under any other equilibrium refinement that leads to the seller's belief of $\theta_2$ dominating the prior in the strict likelihood ratio order after observing an off-path acceptance, and being dominated by the prior in the strict likelihood ratio order after observing an off-path rejection. 
\end{rmk}

\begin{rmk}\label{rmk:negative}
When the buyer's values are negatively correlated across periods, it is possible to sustain $p_A < p_R$ in a PBE-D. Then, consistent with our intuition, the seller may do strictly better in equilibrium than posting all prices in advance. Indeed, the following example shows that with negatively correlated values, the ratchet effect need not appear, and dynamic pricing can be profitable even under limited commitment.
\end{rmk}
\begin{ex}[Negative correlation]\label{ex:negative}
Suppose $T = 2$, $\Theta = \{1, 2\}$ and $\P(\theta_t = 1) = \P(\theta_t = 2) = 0.5$ for both periods. Suppose the types are perfectly anti-correlated. An optimal monopoly price is $2$, and posting this price for both periods generates revenue $2$. Consider instead the following strategy profile. The seller posts $p_1 = 2$ and then posts $p_A = 1$ following an acceptance, $p_R = 2$ following a rejection. Type $(2, 1)$ accepts in both periods. Type $(1, 2)$ rejects in period $1$ and accepts in period $2$. There is no profitable deviation for the buyer. This is also sequentially rational for the seller because she knows $\theta_2 = 2$ after a rejection and $\theta_2 = 1$ after an acceptance. Therefore, a seller-optimal PBE-D gives the seller revenue at least $2.5 > 2$. 
\end{ex}

\section{Extensions} \label{sec:ext}

\subsection{Complements}\label{sec:comps}

Our benchmark model assumes that the buyer's utility in the second period does not depend on the consumption of the first-period product, which rules out the possibility that the products may be complements. This is especially relevant in the setting of dynamic pricing, as one may expect that purchase histories become even more important. 

However, we show that our main result continues to hold even if the products are complements. This is because the ratchet effect becomes stronger precisely when the past information becomes more relevant. When the products are substitutes, we provide an example to show that the ratchet effect need not appear, and the seller may do strictly better in equilibrium than posting all prices in advance (see \Cref{rmk:substitutes} and \Cref{ex:substitutes}). 

To allow for complementarity, we let the distribution of $\theta_2$ depend on the allocation of the first-period product. Specifically, the game proceeds exactly as in \Cref{sec:main}, but $\theta_2 \sim F_{2}( \, \cdot\, | \theta_1, x_1)$ given the realization of $\theta_1$ and the first-period allocation $x_1 \in \{0, 1\}$ (see \Cref{rmk:complements} for an alternative formulation). 

To capture complementarity, we assume that the distribution shifts upward in the sense of likelihood-ratio dominance after consuming the first product. 
\begin{asm}[Complement]\label[manualasm]{asm:comp} For any $\theta_1$, the conditional likelihood ratio $\frac{f(\theta_2| \theta_1, x_1 = 1)}{f(\theta_2| \theta_1, x_1 = 0)}$ is non-decreasing in $\theta_2$.  
\end{asm}

We adapt \Cref{asm:mlrp,asm:Lipschitz,asm:vvfunction}  in  \Cref{sec:main} to this setting. 

\begin{manualasm}{A1'}[MLRP] \label[manualasm]{asm:mlrpx} For any $\theta_1' > \theta_1$ and $x_1 \in \{0, 1\}$, the conditional likelihood ratio $\frac{f(\theta_2| \theta_1', x_1)}{f(\theta_2| \theta_1, x_1)}$ is strictly increasing in $\theta_2$. 
\end{manualasm}

\begin{manualasm}{A2'}[$1/\delta$-Lipschitz] \label[manualasm]{asm:Lipschitzx}
For any $\theta'_1 > \theta_1$,
\[\E[\theta_2|\theta_1', x_1 = 0]- \E[\theta_2|\theta_1, x_1 = 0] < \frac{1}{\delta}(\theta_1'-\theta_1)\,.\]
\end{manualasm}

\begin{manualasm}{A3'}[Regularity] \label[manualasm]{asm:vvfunctionx} The second-period virtual value function 
\[\psi(\theta_1, \theta_2, x_1) := \theta_2 - \frac{1 - F_1(\theta_1)}{f_1(\theta_1)} I(\theta_1, \theta_2, x_1)\]
is non-decreasing in $\theta_1$, $\theta_2$, and $x_1$, where $I(\theta_1, \theta_2, x_1) = - \frac{\partial F_2(\theta_2|\theta_1,  x_1)}{\partial \theta_1} \big \slash f_2(\theta_2|\theta_1,x_1)$.
\end{manualasm}

If the distribution of $\theta_2$ conditional on $\theta_1$ does not depend on $x_1$ as in \Cref{sec:main}, then \Cref{asm:comp} trivially holds and \Cref{asm:mlrpx,asm:Lipschitzx,asm:vvfunctionx} reduce to \Cref{asm:mlrp,asm:Lipschitz,asm:vvfunction}. In \Cref{asm:vvfunctionx}, the virtual value function being non-decreasing in $x_1$ is equivalent to $I(\theta_1, \theta_2, 1) \leq I(\theta_1, \theta_2, 0)$. Intuitively, this condition says that a perturbation in $\theta_1$ would lead to a smaller change in $\theta_2$ if the first product is consumed.

For example, suppose the type process follows an AR(1) process given first-period allocation $x \in \{0, 1\}$: $\theta_2 = \alpha \theta_1 + \epsilon_x$, where $\epsilon_x$ is drawn from a Gaussian distribution $\mathcal{N}(\mu_x, \sigma)$. Then \Cref{asm:comp} is satisfied if $\mu_{1} \geq \mu_{0}$.  \Cref{asm:mlrpx,asm:Lipschitzx}  amount to $0 < \alpha < 1/\delta$.  \Cref{asm:vvfunctionx} is equivalent to that $F_1(\theta_1)$ has monotone hazard rate. 

\begin{theorem}\label{thm:comp}
Under \Cref{asm:mlrpx,asm:Lipschitzx,asm:vvfunctionx,,asm:comp}, 
\begin{itemize}
    \item[(i)] there exists a PBE-D;
    \item[(ii)] in every PBE-D, the seller's revenue is no higher than what she would get by posting prices $p_1, \,p_2$ optimally in advance. 
\end{itemize}
\end{theorem}

\begin{proof}[Proof of \Cref{thm:comp}] See \Cref{app:existence} for the proof of part (i). For expositional convenience, we prove part (ii) for all pure-strategy PBE-D here, and extend the proof to cover all PBE-D in \Cref{app:comp-thm}.  

As in \Cref{subsec:proof}, we first show that in any equilibrium, the buyer follows a threshold rule. For that, it suffices to show that 
\[h(\theta_1) :=  \E[(\theta_2 - p_R)_+ | \theta_1, x_1 = 0]  - \E[(\theta_2 - p_A)_+ | \theta_1, x_1 = 1] \]
satisfies $h(\theta_1') - h(\theta_1) < \frac{1}{\delta}(\theta_1' - \theta_1)$  for any $\theta_1' > \theta_1$, which is proved in \Cref{lem:h3} in \Cref{app:comp-thm}. As shown in \Cref{lem:pricemoncomp} in \Cref{app:comp-thm}, given that the preferences are persistent and the products are complements, after observing acceptance of the first-period offer, the seller believes that the buyer has a higher $\theta_2$ and posts a higher price in the second period (i.e., $p_A \geq p_R$).

Then, as in \Cref{subsec:proof}, using Theorem 1 in \citealt{Pavan2014}, we can provide an upper bound for the seller's revenue by a constrained mechanism design problem with the price and allocation monotonicity constraints:
\begin{align}\label{eq:relax-comp}
    \max_{k, p_A, p_R: p_A \geq p_R} &\E[\varphi(\theta_1)\1_{\theta_1 \geq k}] + \delta \E[\psi(\theta_1, \theta_2, 0)\1_{\theta_1 < k} \1_{\theta_2 \geq p_R}|x_1 = 0]
    \\ \nonumber
    & +  \delta\E[\psi(\theta_1, \theta_2, 1) \1_{\theta_1 \geq k} \1_{\theta_2 \geq p_A} |x_1 = 1] 
     -  \delta \E[(\theta_2 - p_R)_+ |\underline{\theta}_1, x_1 = 0]\,. 
\end{align}
Note that both the virtual value function $\psi(\theta_1, \theta_2, x_1)$ and the distribution of $\theta_2$ depend on the first-period allocation $x_1$.  Let $D_0, \, D_1$ be the boundary curves at which $\psi(\,\cdot\,, \,\cdot\,, 0), \, \psi(\,\cdot\,, \,\cdot\,, 1)$ cross zero respectively, as defined in \Cref{subsec:proof}. Under the regularity condition \Cref{asm:vvfunctionx}, both $D_0(\theta_1)$ and $D_1(\theta_1)$ are non-increasing, and $D_1(\theta_1) \leq D_0(\theta_1)$ for all $\theta_1$. 

For any cutoff $k$, define the effective boundary $D(\theta_1, k) := \1_{\theta_1 < k}D_0(\theta_1) + \1_{\theta_1 \geq k}D_1(\theta_1)$. \Cref{fig:proofcomplements} illustrates. Note that $D(\theta_1, k)$ is non-increasing in $\theta_1$, for any given $k$. Therefore, by the proof of \Cref{claim:key}, there exists an optimal solution to (\ref{eq:relax-comp}) with $p_A = p_R$. 

\begin{figure}[h]
    \centering
	\includegraphics{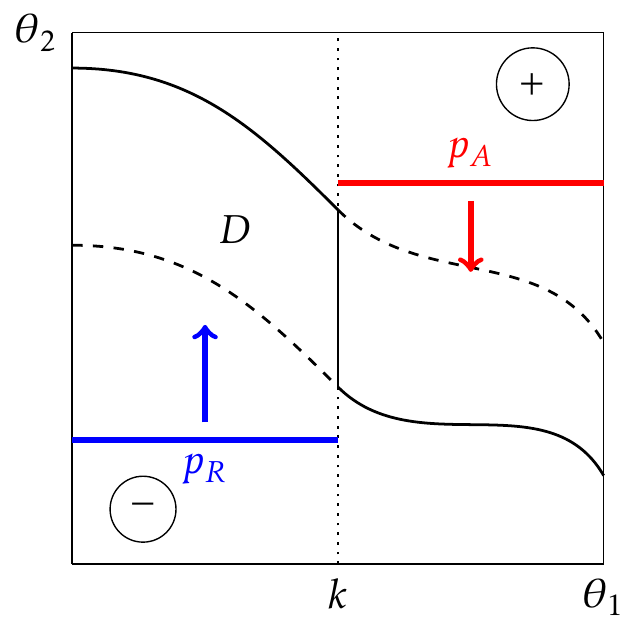}
    \caption{Allocation rule in the type space with the effective boundary}
    \label{fig:proofcomplements}
\end{figure}

Let $(k^*, p^*_2, p^*_2)$ be such an optimal solution to (\ref{eq:relax-comp}).  Consider the mechanism that posts a price $p_1^* = k^* + \delta \E[(\theta_2 - p^*_2)_+|k^*, x_1 = 1] -  \delta \E[(\theta_2 - p^*_2)_+|k^*, x_1 = 0]$ for the first-period item, and a price $p_2^*$ for the second-period item. Note that type $\theta_1 = k^*$ is indifferent between accepting and rejecting the first-period offer. By \Cref{lem:h3}, the first-period offer is accepted if and only if $\theta_1 \geq k^*$. Clearly, the second-period offer would be accepted if and only if $\theta_2 \geq p^*_2$. Thus, this mechanism implements the allocation rule $x_1(\theta_1) = \1_{\theta_1 \geq k^*}, \, x_2(\theta_1, \theta_2) = \1_{\theta_2 \geq p_2^*}$ that solves ($\ref{eq:relax-comp}$). Moreover, the lowest type gets an expected payoff $V(\underline{\theta}_1) = \delta \E[(\theta_2 - p^*_2)_+ |\underline{\theta}_1, x_1 = 0]$. Therefore, by the same argument as in \Cref{subsec:proof}, the seller's revenue in any PBE-D is no higher than what she would get from posting $p_1,\, p_2$ optimally in advance.  
\end{proof}

\begin{rmk}\label{rmk:complements}
Instead of letting the distribution of $\theta_2$ depend on $x_1$, one may also introduce complementary products by assuming a non-additively separable utility for the buyer, $\theta_1 x_1 + \delta (\theta_2 x_2 + \kappa(\theta_1, \theta_2) x_1 x_2)$, where $\kappa$ governs the extent of complementarity. Note that this is equivalent to our formulation. To see the equivalence, define $\tilde{\theta}_2 = \theta_2$ if $x_1 = 0$ and $\tilde{\theta}_2 = \theta_2 + \kappa(\theta_1, \theta_2)$ if $x_1 = 1$. Then, the buyer's utility is additively separable in $(\theta_1, \tilde{\theta}_2)$ and the distribution of $\tilde{\theta}_2$ depends on $\theta_1$ and $x_1$. 
\end{rmk}

\begin{rmk}\label{rmk:substitutes}
If the two items are substitutes, then consuming the first item decreases the demand for the second item. In that case, it is possible to sustain repeated purchase discounts (i.e.,  $p_A < p_R$) in equilibrium. Then, consistent with our intuition, the seller may do strictly better in equilibrium than posting all prices in advance, as shown by the following example. 
\end{rmk}

\begin{ex}[Substitutes]\label{ex:substitutes}
Suppose $\Theta = \{1,2\}$ and $\mathbb{P}(\theta_t=1) = \mathbb{P}(\theta_t=2)=0.5$ for both periods. Suppose that types are perfectly correlated and that the buyer's utility is given by $\theta_1 x_1 + \theta_2 x_2 -1.5 x_1 x_2$ (see \Cref{rmk:complements}). Since $\kappa(\theta_1,\theta_2) = -1.5$, the two goods are substitutes. If the seller posts $p_1, \, p_2$ in advance, then the seller can obtain at most revenue $1$ (by posting $p_1 = p_2 = 1$). Consider instead the following strategy. The seller posts a price $p_1 = 1$ in the first period, then posts a price $p_{A} = 0.5$ after an acceptance and price $p_R = 1$ after a rejection. Both types accept in the first period, and the high type also accepts in the second period. Let the seller's belief be concentrated at $\theta_1 = 1$ after an off-path rejection. This strategy is then sequentially rational for the seller. Therefore, a seller-optimal PBE-D gives the seller revenue at least $1.25 > 1$. 
\end{ex}

\subsection{Multiple Periods}\label{sec:multi-mp}

Our benchmark model assumes that $T = 2$. One natural extension would call for a similar revenue bound for any finite time horizon $T$. With $T > 2$, however, it is known from \citet{devanur2019perfect} that a pure-strategy threshold equilibrium generally does not exist with perfectly correlated types (see Theorem 1 in \citealt{devanur2019perfect}). 

To proceed, we introduce the following multiple-period game. In each period $t \in \{1, \dots, T\}$, the seller can either set the price $p_t$ for period $t$ only, or (publicly) commit to a sequence of prices $\{p_t, \dots, p_T\}$ for the unsold items. The buyer then decides whether to purchase the item as in \Cref{sec:main}. 

When $T = 2$, the modification reduces to providing the seller an option to post all prices in advance. So our main result is equivalent to that there exists a PBE-D of the modified game in which the seller simply posts all prices in advance.

We show that this statement (in fact, a stronger statement) extends to any finite time horizon $T$, under the following distributional assumption:  

\begin{asm}[Log-concave AR(1)]\label[manualasm]{asm:ar}
The first-period type $\theta_1$ follows a log-concave distribution $F_1$ (with bounded support). For all $t \geq 2$, 
\[\theta_t = \alpha_t \theta_{t-1} + \epsilon_{t}\]
where  $\alpha_t \in (0, \frac{1}{2\delta})$,  and $\epsilon_t$ is independently drawn from a log-concave distribution $G_t$ (with bounded support). 
\end{asm}

Recall that we say two equilibria
are outcome-equivalent if they almost surely induce the same allocation and the same total discounted payment.

\begin{theorem}\label{thm:multi} Under \Cref{asm:ar}, for any finite $T$, 
\begin{itemize}
    \item[(i)]  there exists a PBE-D of the modified game;
    \item[(ii)] every PBE-D is outcome-equivalent to one in which the seller commits to $\{p^*_1, \cdots, p^*_T\}$, where $p^*_t$ is the monopoly price with respect to the marginal distribution $F_t$.
\end{itemize}
\end{theorem}

\begin{proof}[Proof of \Cref{thm:multi}] See \Cref{app:existence-mp} for the proof of part (i). We now prove part (ii).

In \Cref{lem:threshold-mp} in \Cref{app:multi-thm}, we show that at any history $\hat{h}_t$, the buyer always adopts a threshold rule: $x_t(\hat{h}_t) = \1_{\theta_t \geq k_t}$, where the threshold $k_t$ depends only on the public history. This statement is proved inductively by bounding the derivative of the buyer's continuation payoffs following an acceptance or a rejection and then showing that an appropriate single-crossing condition holds. 

The rest of the proof proceeds by induction on $T$. The base case $T = 1$ is trivial. For the inductive step, fix any PBE-D. Consider what action the seller decides to take in period $2$. Because the buyer adopts a threshold rule in period $1$ (say the cutoff type is $k$), the seller's posterior belief about $\theta_1$ follows a truncated distribution of $F_1$ after observing an acceptance or a rejection. Because log-concavity is preserved under truncation, convolution, and linear operations, her posterior belief about $(\theta_{2}, \theta_{3}, \dots, \theta_{T})$ in period $2$ also satisfies \Cref{asm:ar}.\footnote{When $k = \underline{\theta}_1$ (or $\overline{\theta}_1$), as noted in \Cref{rmk:d1} and shown in \Cref{app:d1}, it is without loss of generality to let the seller's belief of $\theta_1$ be $\underline{\theta}_1$ (or $\overline{\theta}_1$) when observing an off-path rejection (or acceptance).} Moreover, because the buyer uses a threshold strategy, information about $\theta_1$ does not matter to both players. Invoking the inductive hypothesis then implies that any equilibrium of the subgame must be outcome-equivalent to the one where the seller posts the monopoly prices from period $2$ onward. 

Let $\{p^{A}_{t}\}_{t=2}^T$ denote the sequence of monopoly prices with respect to the marginals $\{F_{t}(\, \cdot\, | \theta_1 \geq k)\}_{t=2}^{T}$ and similarly define  $\{p^{R}_{t}\}_{t=2}^T$. Because the posterior distributions are all log-concave, both $\{p^{A}_{t}\}_{t=2}^T$ and  $\{p^{R}_{t}\}_{t=2}^T$ are uniquely defined. Moreover, because these prices are unique, it is without loss of generality to consider only pure strategies. Note that under \Cref{asm:ar}, $(\theta_1, \theta_t)$ satisfies an appropriate MLRP condition (see \Cref{lem:MLRP-mp} in \Cref{app:multi-thm}). Therefore, we have $p^{A}_t \geq p^{R}_t$ for all $t$ (see \Cref{lem:pricemon-mp} in \Cref{app:multi-thm}).  This enables us to adopt the same approach as in the proof of \Cref{thm:main}. In particular, using Theorem 1 of \citet{Pavan2014}, we have that the following relaxed problem gives an upper bound for the seller's revenue: 
\begin{equation}\label{eq:relax_multi}
\max_{(x, t) \in \X}  \E \Big [ \sum_{t=1}^T  \delta^{t-1} \psi_t(\theta_1, \dots, \theta_t)x_t(\theta_1, \dots, \theta_t) \Big ] - V_{x,t}(\underline{\theta}_{1})
\end{equation}
where $\mathcal{X}$ encodes that the seller selects a price mechanism subject to the threshold allocation constraint for period $1$ and price monotonicity constraint for period $t \geq 2$, and $\psi_t$ is the period-$t$ virtual value function (see Section 4 of \citealt{Pavan2014}). In particular, under \Cref{asm:ar}, we have
\[ \psi_t(\theta_1, \dots, \theta_t) =  \theta_t - \frac{1 - F_1(\theta_1)}{f_1(\theta_1)} \prod_{s=2}^t \alpha_{s}\] 
(see Example 2 in \citealt{Pavan2014}). By the same argument as in \Cref{subsec:proof}, the problem reduces to   
\begin{align*}
\max_{\substack{ k, \{p^{A}_{t}\}_{t=2}^T, \{p^{R}_{t}\}_{t=2}^T \\ p^A_t \geq p^R_t ,\, \forall t \geq 2 } }   \E[\varphi(\theta_1) \1_{\theta_1 \geq k} ] + \sum_{t=2}^T \delta^{t-1} \Big( \E[\psi_t(\theta_1, \theta_t) &(\1_{\theta_1 \geq k} \1_{\theta_t \geq p^A_t} + \1_{\theta_1 < k} \1_{\theta_t \geq p^R_t})] -\E[(\theta_t - p^{R}_t)_+|\underline{\theta}_{1}] \Big)\,. 
\end{align*}

For any $k$, the $t$-th term in the above sum is of the same form as in (\ref{eq:star}). For any $t \geq 2$, note that $\psi_t(\theta_1, \theta_t)$ is strictly increasing in $\theta_t$ and non-decreasing in $\theta_1$ because log-concavity implies monotone hazard rate. Thus, by \Cref{lem:support} in \Cref{app:multi-thm},  if $(p^A_t,\, p^R_t)$ induces a period-$t$ allocation rule that is not almost surely equal to $\1_{\theta_t \geq p_t}$ for some $p_t$, then the seller does strictly worse for the period-$t$ revenue. Then, for any $k$, the outcome under any optimal choice of $\big(\{p^A_t\}_{t=2}^T,\, \{p^R_t\}_{t=2}^T\big)$ must be equivalent to the outcome under some price path $\{p_2, \dots, p_T\}$. Therefore, for the $T$-period game, every PBE-D must be outcome-equivalent to one in which the seller commits to $\{p^*_1,\dots,p^*_T\}$. This proves the inductive step.
\end{proof}

\section{Conclusion}\label{sec:disc}

This paper studies dynamic pricing with persistent private information and limited commitment. In a two-period game, the main result shows that regardless of the degree of persistence of the private information, any equilibrium under a D1-style refinement gives the seller revenue no higher than what she would get from posting all prices in advance. The ability to condition prices on purchase histories, which are informative about the buyer's values, introduces the ratchet effect that is strong enough to outweigh the gains from learning.  

This result contrasts with the case of full commitment --- a large literature on dynamic mechanism design highlights the power of dynamically responding to the buyer's evolving private information. As our analysis shows, to benefit from dynamic pricing, the seller must commit to a lower price precisely when he \textit{ex post} knows that the buyer has a higher willingness to pay. That is, dynamic rent extraction always goes in the opposite direction of the seller's sequential rationality, and thus hinges critically on her long-term commitment power. Building on this insight, we also show how similar results hold when the products are complements or when there are multiple periods. In contrast, when the products are substitutes or when the values are negatively correlated, the ratchet effect need not appear. Then, as we show by example, dynamic pricing can be effective even under limited commitment.

\begin{appendices}
\crefalias{subsection}{appendix}
\crefalias{section}{appendix}

\section{Omitted Proofs}\label{app:proofs}
\subsection{Proof of \Cref{lem:MLRPone}}\label{app:lem1}
 Fix any $k \in (\underline{\theta}_1, \overline{\theta}_1)$. Note that 
    \[\frac{f(\theta_2 | \theta_1 \geq k) }{f(\theta_2 | \theta_1 < k) }\propto \frac{ \int_{k}^{ \overline{\theta}_1} f(\theta_2| \theta_1 = s') f(\theta_1 = s')ds' }{ \int_{\underline{\theta}_1}^k f(\theta_2| \theta_1 = s) f(\theta_1 = s)ds }\,.\]
It then suffices to show that for any $s' > k$, 
\[\frac{f(\theta_2| \theta_1 = s')}{ \int_{\underline{\theta}_1}^k f(\theta_2| \theta_1 = s) f(\theta_1 = s)ds} \]
is strictly increasing in $\theta_2$. Equivalently, 
\[\frac{\int_{\underline{\theta}_1}^k f(\theta_2| \theta_1 = s) f(\theta_1 = s)ds}{ f(\theta_2| \theta_1 = s')} \]
is strictly decreasing in $\theta_2$. It then suffices to show
\[\frac{f(\theta_2| \theta_1 = s)}{ f(\theta_2| \theta_1 = s')} \]
is strictly decreasing in $\theta_2$ for every $s'>s$, but that is implied by the MLRP property.

\subsection{Proof of \Cref{lem:MLRPtwo}}\label{app:lem2}
It is clear that 
\[ p_i \in \argmax_{p}[ \log(p) + \log(1 - F_i(p))] \,.\]
The claim follows from the monotone selection theorem, if we show that $v(p, i):=\log(1 - F_i(p))$ has strictly increasing differences.  Let \[\Delta(p) := \log(1 - F_2(p)) - \log(1 - F_1(p))\,. \]
Note that
\[\Delta'(p) =  \frac{f_1(p)}{1 - F_1(p)} - \frac{f_2(p)}{1 - F_2(p)} > 0\]
since strict likelihood-ratio dominance implies strict hazard-rate dominance.

\subsection{Completion of Proof of Part (ii) of \Cref{thm:main}}\label{app:main}

We complete the proof of part (ii) of \Cref{thm:main} by extending the proof in the main text to cover all PBE-D. Instead of choosing two fixed prices $p_A, p_R$, the seller now chooses two random variables $\tilde{p}_A, \tilde{p}_R$. We refer to their realizations as $p_A, p_R$. (It is without loss of generality to only focus on mixing over $p_A, p_R$ since mixing over $p_1$ implies that there always exists a deterministic $p_1$ achieving the same revenue.) 

\paragraph{Step 1:} We show that the first-period allocation rule is still a threshold rule, and that price monotonicity holds in the sense that $p_A \geq p_R$ for any realized $p_A, p_R$, i.e., $\min\{ \supp(\tilde{p}_A) \} \geq \max\{\supp(\tilde{p}_R)\}$. Note that (\ref{eq:h}) is preserved under averaging. So if we define 
\[h(\theta_1) := \E_{\tilde{p}_A, \tilde{p}_R}[\E[(\theta_2 - p_R)_{+}- (\theta_2 - p_A)_{+} | \theta_1]]\]
then we still have $h(\theta_1') - h(\theta_1) < \frac{1}{\delta}(\theta_1' - \theta_1)$ for any $\theta_1' > \theta_1$. Type $\theta_1$ accepts if
\[(\theta_1 - p_1) + \delta \E_{\tilde{p}_A}[\E[(\theta_2 - p_A)_{+}|\theta_1]] \geq \delta \E_{\tilde{p}_R}[\E[(\theta_2 - p_R)_{+}|\theta_1]]    \,.
\]
Then as before, $g(\theta_1) := \frac{1}{\delta}(\theta_1 - p_1)$ crosses $h$ at most once from below. Therefore, the first-period allocation is characterized by a cutoff type $k$.

Suppose $k \in (\underline{\theta}_1, \overline{\theta}_1)$. By \Cref{lem:MLRPone}, $F_2( \, \cdot\, | \theta_1 \geq k) \succ_L F_2( \, \cdot\, | \theta_1 < k)$. Note that any realized $p_A$ must be an optimal price under belief $F_2( \, \cdot\, | \theta_1 \geq k)$, as the seller only mixes among optimal prices. Similarly, any realized $p_R$ must be an optimal price under belief $F_2( \, \cdot\, | \theta_1 < k)$. Hence $p_A \geq p_R$ for any realized $p_A,\, p_R$ by \Cref{lem:MLRPtwo}. The case of $k = \underline{\theta}_1$ or $\overline{\theta}_1$ is handled similarly as in the proof of \Cref{prop:pricemon}.

\paragraph{Step 2: } The same relaxation argument applies. The only difference is that the seller now commits to a randomized price mechanism in the second period with the constraint \[\min\{ \supp(\tilde{p}_A) \} \geq \max\{\supp(\tilde{p}_R)\}\]
denoted by $\tilde{p}_A \succeq \tilde{p}_R$. The same steps go through, except that the relaxed problem is now 
\[
\max_{k, \tilde{p}_A, \tilde{p}_R: \tilde{p}_A \succeq \tilde{p}_R} \E[\varphi(\theta_1)\1_{\theta_1 \geq k}] + \delta \E[\psi(\theta_1, \theta_2) (\1_{\theta_1 \geq k} \1_{\theta_2 \geq \tilde{p}_A}+ \1_{\theta_1 < k} \1_{\theta_2 \geq \tilde{p}_R}) ] - \delta \E[(\theta_2 - \tilde{p}_R)_+ | \theta_1 = \underline{\theta}_1]
\]
where $\tilde{p}_A, \tilde{p}_R$ are two random variables. 

\paragraph{Step 3:} To solve this, we consider a further relaxation. We couple $\tilde{p}_A, \tilde{p}_R$ in a common probability space $(\Omega, \F, Q)$. Consider the requirement that  $\tilde{p}_A(\omega) \geq \tilde{p}_R(\omega)$ for all $\omega \in \Omega$, denoted by $\tilde{p}_A \geq \tilde{p}_R$. The disjoint support constraint clearly implies this, so this is a relaxation of the problem above. The seller's problem is to maximize over $k$ and two $\F$-measurable functions:
\begin{align*}
    \max_{k, \tilde{p}_A, \tilde{p}_R: \tilde{p}_A \geq \tilde{p}_R} \E_{Q}\Big[\E[\varphi(\theta_1)\1_{\theta_1 \geq k}] + \delta \E[\psi(\theta_1, \theta_2) (\1_{\theta_1 \geq k} \1_{\theta_2 \geq \tilde{p}_A(\omega)}&+ \1_{\theta_1 < k} \1_{\theta_2 \geq \tilde{p}_R(\omega)}) ]  \\
    &- \delta \E[(\theta_2 - \tilde{p}_R(\omega))_+ | \theta_1 = \underline{\theta}_1] \Big] \,.
\end{align*}
Now fix any $\omega$. Let $p_A, p_R$ denote $ \tilde{p}_A(\omega), \tilde{p}_R(\omega)$. Then the ex post maximization is 
\[
\max_{k, p_A, p_R: p_A \geq p_R} \E[\varphi(\theta_1)\1_{\theta_1 \geq k}] + \delta \E[\psi(\theta_1, \theta_2) (\1_{\theta_1 \geq k} \1_{\theta_2 \geq p_A}+ \1_{\theta_1 < k} \1_{\theta_2 \geq p_R}) ] - \delta \E[(\theta_2 - p_R)_+ | \theta_1 = \underline{\theta}_1]
\]
which is exactly what we have in \Cref{subsec:proof} and has a maximizer $(k^*, p^*_2, p^*_2)$, irrespective of $\omega$. Hence, there is a solution to the ex ante maximization problem (with the constraint $\tilde{p}_A \geq \tilde{p}_R$) with $\tilde{p}_A  = \tilde{p}_R = p^*$. This solution also satisfies the constraint $\tilde{p}_A \succeq \tilde{p}_R$ and thus solves the mechanism design relaxation where the seller is allowed to use randomized price mechanisms in the second period. So the seller's optimal constrained mechanism is still to simply post two deterministic prices in advance. Therefore, her revenue in any PBE-D is no higher than $\sum_t \delta^{t-1} \pi^*_t$. 

\subsection{Completion of Proof of Part (ii) of \Cref{thm:comp}}\label{app:comp-thm}
\begin{lemma}\label{lem:h3} In every pure-strategy PBE-D, 
$h(\theta'_1) - h(\theta_1) < \frac{1}{\delta}(\theta_1' - \theta_1)$ for any $\theta_1' > \theta_1$.  
\end{lemma}
\begin{proof}
Fix any $\theta'_1 > \theta_1$. Integration by parts implies
\begin{align*}
    \E[(\theta_2 - p_R)_+ | \theta'_1, x_1 = 0] &- \E[(\theta_2 - p_R)_+ | \theta_1, x_1 = 0] \\
    &= \int_0^\infty [\P(\theta_2 > s | \theta'_1, x_1 = 0) - \P(\theta_2 > s | \theta_1, x_1 = 0)] \1_{s \geq p_R} ds  \\
     &\leq  \int_0^\infty [\P(\theta_2 > s | \theta'_1, x_1 = 0) - \P(\theta_2 > s | \theta_1, x_1 = 0)]  ds \\
     &= \E[\theta_2 |\theta'_1, x_1 = 0] - \E[\theta_2 | \theta_1, x_1 = 0] < \frac{1}{\delta}(\theta'_1 - \theta_1) 
\end{align*}
where we used FOSD (implied by \Cref{asm:mlrpx}) and \Cref{asm:Lipschitzx}. The claim follows by noting that FOSD implies
$\E[(\theta_2 - p_A)_+ | \theta'_1, x_1 = 1] -  \E[(\theta_2 - p_A)_+ | \theta_1, x_1 = 1] \geq 0$
\end{proof}

\begin{lemma}\label{lem:pricemoncomp}
In every pure-strategy PBE-D,  $p_A \geq p_R$.
\end{lemma}
\begin{proof}
Using the assumptions, we have \[F_2(\, \cdot\,  | \theta'_1, x_1 = 1) \succ_L F_2(\, \cdot\,  | \theta_1, x_1 = 1) \succeq_L  F_2(\, \cdot\, | \theta_1, x_1 = 0)\]
for any $\theta'_1 > \theta_1$, where $\succeq_L$ denotes (weak) likelihood-ratio dominance. This then implies that for any $k \in (\underline{\theta}_1, \overline{\theta}_1)$, $F_2(\, \cdot\,  | \theta_1 \geq k, x_1 = 1) \succ_L F_2(\, \cdot\,  | \theta_1 < k, x_1 = 0)$ by the proof of \Cref{lem:MLRPone}. Price monotonicity then follows from \Cref{lem:MLRPtwo}. When $k = \underline{\theta}_1$ or $\overline{\theta}_1$, the same holds by the D1 criterion as in the proof of \Cref{prop:pricemon}. 
\end{proof}

\noindent \textbf{Extension to all PBE-D.} We follow the same notation and the same steps as in \Cref{app:main}. For Step 1, note that \Cref{lem:h3} continues to hold if we define \[h(\theta_1) := \E_{\tilde{p}_A, \tilde{p}_R}\big [\E[(\theta_2 - p_R)_+ | \theta_1, x_1 = 0]  - \E[(\theta_2 - p_A)_+ | \theta_1, x_1 = 1] \big ]\,.\]
Moreover, by the same reasoning as in \Cref{app:main} and the proof of \Cref{lem:pricemoncomp}, we have $p_A \geq p_R$ for any realized $p_A, \, p_R$. Steps 2 and 3 proceed in exactly the same way as in \Cref{app:main}, and are thus omitted.

\subsection{Completion of Proof of Part (ii) of \Cref{thm:multi}} \label{app:multi-thm}

\begin{lemma} \label{lem:threshold-mp}
For any seller's strategy, the buyer's optimal strategy is given by a threshold rule: at every history $\hat{h}_t$,  $x_t(\hat{h}_t) = \1_{\theta_t \geq k_t}$, where $k_t$ depends only on the public history.   
\end{lemma} 

\begin{proof}
Note that at any history $\hat{h}_t$, the buyer solves the following problem 
\[\max_{x \in \{0, 1\}} [(\theta_t - p_t) + U^{A}_t(\theta_t, h^\dagger_t)]x + U^{R}_t(\theta_t, h^\dagger_t)(1 - x)\]
where $U^A_t, U^R_t$ are the continuation payoffs following an acceptance and a rejection, and we decompose the history of the buyer $\hat{h}_t$ into $(\theta_t,h^\dagger_t)$. 

We prove the following three claims together by backward induction: 

\begin{itemize}
    \item[(1)] $x_t(\hat{h}_t) = \1_{\theta_t \geq k_t}$, where $k_t$ depends only on the public history. 
    \item[(2)] $U^A_t(\cdot, h^\dagger_t), U^R_t(\cdot, h^\dagger_t)$  are  1-Lipschitz and non-decreasing functions of $\theta_t$.
    \item[(3)] $U^A_t(\theta_t, \cdot), U^R_t(\theta_t, \cdot)$ depend on $h^\dagger_t$ only through public history $h^{o}_t := (h_t, p_t)$. 
\end{itemize}

\noindent \textbf{Base case:} In the last period, $k_T = p_T$ and $U^A_T = U^R_T = 0$. So the claims hold.

\noindent \textbf{Inductive step:} Note that 
\[U^A_t(\theta_t, h^\dagger_t) = \delta \E \Big[ \max\Big\{\theta_{t+1}-p_{t+1} + U^{A}_{t+1}(\theta_{t+1}, h^\dagger_{t+1}), U^{R}_{t+1}(\theta_{t+1}, h^\dagger_{t+1})\Big\} \Big |\theta_t,  h^\dagger_t, x_t = A \Big]\,.\]
By definition, $p_{t+1}$ depends only on the public history. Since the type process is Markovian, the distribution of $\theta_{t+1}$ depends only on $\theta_t$. By the inductive hypothesis, $U^{A}_{t+1}(\theta_{t+1}, \cdot)$, $U^{R}_{t+1}(\theta_{t+1}, \cdot)$ depend only on the public history. Therefore, for any fixed $\theta_t$, $U^A_t(\theta_t, \cdot)$ depends on $h^\dagger_t$ only through the public history $h^o_t$. The same argument works for $U^R_t(\theta_t, \cdot)$.  Moreover, by the inductive hypothesis, for any fixed $p$, 
\[\max\Big\{\theta_{t+1}-p + U^{A}_{t+1}(\theta_{t+1}, h^o_{t+1}), U^{R}_{t+1}(\theta_{t+1}, h^o_{t+1})\Big\}\]
is $2$-Lipschitz and non-decreasing in $\theta_{t+1}$. Therefore, 
\[U_{t+1}(\theta_{t+1}, h^o_{t}): = \E_{p_{t+1}\sim \sigma }\Big[\max\Big\{\theta_{t+1}-p_{t+1} + U^{A}_{t+1}(\theta_{t+1}, h^o_{t+1}), U^{R}_{t+1}(\theta_{t+1}, h^o_{t+1})\Big\} \Big | h^o_t,  x_t = A\Big]\]
is also $2$-Lipschitz and non-decreasing in $\theta_{t+1}$, where $\sigma$ is the seller's strategy.  Fix any $\theta'_t > \theta_t$, integration by parts yields that 
\[U^A_t(\theta'_t , h^o_t)  - U^A_t(\theta_t,  h^o_t)  = \delta \int_0^\infty \big [\P(\theta_{t+1} > s | \theta'_t) - \P(\theta_{t+1} > s | \theta_t) \big ] U_{t+1}'(s,  h^o_t) ds \,.\]
Under \Cref{asm:ar}, $\P(\theta_{t+1} > s | \theta'_t) \geq \P(\theta_{t+1} > s | \theta_t)$. Thus, 
\begin{align*}
    U^A_t(\theta'_t, h^o_t)  - U^A_t(\theta_t, h^o_t) &\leq  2\delta \int_0^\infty \big [\P(\theta_{t+1} > s | \theta'_t) - \P(\theta_{t+1} > s | \theta_t) \big ] ds   \\
    &= 2\delta (\E[\theta_{t+1}|\theta'_t] - \E[\theta_{t+1} | \theta_t]) < \theta'_t - \theta_t
\end{align*}
where the last inequality is due to the assumption $\alpha_t \in (0, \frac{1}{2\delta})$ for all $t$. This shows that $U^A_t(\cdot, h^o_t)$ is 1-Lipschitz. It is clear from the above that $U^A_t(\cdot, h^o_t)$ is non-decreasing. The same arguments work for $U^R_t(\cdot, h^o_t)$. Together these also imply that for any $h^\dagger_t$, 
\[H(\theta_t, h^\dagger_t) := (\theta_t - p_t) + U^{A}_t(\theta_t, h^\dagger_t) -  U^{R}_t(\theta_t, h^\dagger_t) = (\theta_t - p_t) + U^{A}_t(\theta_t, h^o_t) -  U^{R}_t(\theta_t, h^o_t)\]
is a strictly increasing function and thus crosses $0$ at most once from below. Therefore, the buyer's strategy is given by a threshold rule and the threshold $k_t$ depends only on the public history $h^o_t$. This proves the inductive step. 
\end{proof}

\begin{lemma}\label{lem:MLRP-mp}
For any $2 \leq t \leq T$ and any $\theta_1' > \theta_1$, the conditional likelihood ratio $\frac{f(\theta_t | \theta_1')}{f(\theta_t | \theta_1)}$
is non-decreasing in $\theta_t$.
\end{lemma}
\begin{proof}
Fix any $t$. We can write 
\[\theta_t =  \alpha_t \theta_{t-1}  + \epsilon_t =  \alpha_t (\alpha_{t-1} \theta_{t-2} + \epsilon_{t-1}) + \epsilon_t = \cdots =  \Big(\prod_{s=2}^{t}\alpha_s\Big) \theta_1 + \sum_{s = 2}^t \Big(\prod_{n = s+1}^{t} \alpha_n \Big) \epsilon_s \]
where we interpret $\prod_{n = s+1}^{t} \alpha_n = 1$ when $s = t$. Let $\alpha := \prod_{s=2}^{t} \alpha_s > 0$. Since $\epsilon_s$ are mutually independent and each follows a log-concave distribution, $\epsilon := \sum_{s = 2}^t (\prod_{n = s+1}^{t} \alpha_n) \epsilon_s$ follows a log-concave distribution. Moreover, $\epsilon$ and $\theta_1$ are independent. Therefore $\theta_t =  \alpha \theta_1 + \epsilon$ where $\alpha > 0$ and $\epsilon$ follows some log-concave distribution $G$ (with a density $g$).  Since $G$ is a log-concave distribution, the conditional distribution of $\theta_t$ has MLRP in $\alpha \theta_1$ by standard results, i.e.,  
\[\frac{f_t(\theta_t | \theta_1')}{f_t(\theta_t | \theta_1)} = \frac{g(\theta_t -  \alpha \theta'_1 )}{g(\theta_t -  \alpha \theta_1 )}\]
is non-decreasing in $\theta_t$, for any $\alpha \theta'_1 > \alpha \theta_1$. Since $\alpha > 0$, the latter is simply $\theta'_1 > \theta_1$, proving the claim.
\end{proof}

\begin{lemma}\label{lem:pricemon-mp} For any $2 \leq t \leq T$, $p^A_t \geq p^R_t$. 
\end{lemma}
\begin{proof}
Because $F_t(\, \cdot\, |\theta_1 < k)$ and $F_t(\, \cdot\, |\theta_1 \geq k)$ are both log-concave, both admit a unique optimal monopoly price.\footnote{To see this, note that if $F$ is log-concave, then $\log(p) + \log(1 - F(p))$ is strictly concave and admits a unique maximizer.} Then, we no longer need the monotone selection theorem used in \Cref{lem:MLRPtwo}. It suffices to show that $F_t(\, \cdot\, |\theta_1 \geq k)$ weakly dominates $F_t(\, \cdot\,  | \theta_1 < k)$ in hazard rate order, which is implied by weak likelihood-ratio dominance $F_t(\, \cdot\, | \theta_1 \geq k) \succeq_L F_t(\, \cdot\,  | \theta_1 < k)$. By the proof of \Cref{lem:MLRPone}, it suffices to show $(\theta_1, \theta_t)$ satisfies the weak MLRP condition, which is established in \Cref{lem:MLRP-mp}. The case of $k = \underline{\theta}_1$ or $\overline{\theta}_1$ follows by a similar argument as in the proof of \Cref{prop:pricemon}. 
\end{proof}

\begin{lemma} \label{lem:support} Suppose $\psi(\theta_1, \theta_2)$ is non-decreasing in $\theta_1$ and strictly increasing in $\theta_2$. Then, for any $k \in [\underline{\theta}_1, \overline{\theta}_1]$, and any optimal solution $(p_A, p_R)$ to the following problem 
\[ \max_{p_A, p_R: p_A \geq p_R} \E[\psi(\theta_1, \theta_2) (\1_{\theta_1 \geq k} \1_{\theta_2 \geq p_A} + \1_{\theta_1 < k} \1_{\theta_2 \geq p_R})] -\E[(\theta_2 - p_R)_+|\underline{\theta}_{1}]\,,\]
there exists an optimal solution $(\tilde{p}_A, \tilde{p}_R)$ to the same problem such that (i) $\tilde{p}_A = \tilde{p}_R$ and (ii) \[\1_{\theta_1 \geq k}\1_{\theta_2 \geq p_A} + \1_{\theta_1 < k}\1_{\theta_2 \geq p_R} \eqas \1_{\theta_1 \geq k}\1_{\theta_2 \geq \tilde{p}_A} + \1_{\theta_1 < k}\1_{\theta_2 \geq  \tilde{p}_R}\,.\] 
\end{lemma}
\begin{proof}
Let $D(\theta_1)$ be defined as in \Cref{subsec:proof} and $z := D(k)$. Note that $D(\cdot)$ is non-increasing. Fix any optimal solution with $p_A > p_R$. Suppose $p_R < z < p_A$. Then, the set of types $[\underline{\theta}_1, k] \times [p_R, z]$ must have measure $0$ because otherwise increasing $p_R$ up to $z$ strictly increases the objective. Similarly, the set of types $[k, \overline{\theta}_1] \times [z, p_A]$ must also have measure $0$. Thus, the second-period allocation rule $\1_{\theta_1 \geq k}\1_{\theta_2 \geq p_A}  +  \1_{\theta_1 < k}\1_{\theta_2 \geq p_R} \eqas \1_{\theta_2 \geq z}$. The same argument holds for the other two cases $z \leq p_R < p_A$ and $p_R < p_A \leq z$ as well.
\end{proof}

\subsection{Proof of Part (i) of \Cref{thm:main} and Part (i) of \Cref{thm:comp}}
\label{app:existence}

As explained in \Cref{rmk:d1} and shown formally in \Cref{app:d1}, it suffices to show the existence of a PBE with a particular class of off-path beliefs. 
\begin{defn}\label{defn:PBEstar}
A PBE-$^\star$ is a perfect Bayesian equilibrium with the following restrictions:
\begin{itemize}
    \item After observing an off-path rejection in period $t$ under belief $\mu_t$ on $\theta_t$, the seller believes that the period-$t$ type of the buyer is $\min \{\supp(\mu_t)\}$.
    \item After observing an off-path acceptance in period $t$ under belief $\mu_t$ on $\theta_t$, the seller believes that the period-$t$ type of the buyer is $\max \{\supp(\mu_t)\}$.
\end{itemize}
\end{defn}

In what follows, we show the existence of a PBE-$^\star$ for the game described in \Cref{sec:comps}, which then implies part (i) of \Cref{thm:main} and part (i) of \Cref{thm:comp}.  

Consider the following auxiliary complete information one-shot game parameterized by $p$ with $3$ players: $B$, $S_A$, $S_R$.
\noindent
Player $B$ chooses $k \in [\underline{\theta}_1, \overline{\theta}_1]$ to maximize
\[\E[\{(\theta_1 - p) + \delta (\theta_2 - p_A)_+\} \1_{\theta_1 \geq k}|x_1 = 1] + \E[\delta (\theta_2 - p_R)_+ \1_{\theta_1 \leq k}|x_1 = 0]\,.\]
Player $S_A$ chooses $p_A \in [\underline{\theta}_2, \overline{\theta}_2]$ to maximize 
\[ (1 - F_2(p_A | \theta_1 \geq k, x_1 = 1)) p_A \,.\]
Player $S_R$ chooses $p_R \in [\underline{\theta}_2, \overline{\theta}_2]$ to maximize
\[ (1 - F_2(p_R | \theta_1 \leq k, x_1 = 0)) p_R \,.\]
Note that this is a continuous game. By \citet{Glicksberg1952}, there exists a mixed-strategy Nash equilibrium, denoted by measures $(\eta^p_B, \eta^p_A, \eta^p_R)$. 

By \Cref{lem:h3}, we know that Player $B$ will not mix. Thus we can identify the equilibrium strategy $\eta^p_B$ by a cutoff $k^p$. For any fixed $p$, let $\mathcal{E}^p$ be the set of equilibrium tuples $(k^p, \pi^p_A, \pi^p_R)$, where $\pi^p_A,\pi_R^p$ are equilibrium payoffs of players $S_A$ and $S_R$. We claim that $\mathcal{E}^p$ is a closed and bounded set in $\R^{3}$. Clearly, it is bounded. To show it is closed, take any sequence
$(k^p_n, \pi^p_{A, n}, \pi^p_{R, n}) \in \mathcal{E}^p$ converging to  $(k^p, \pi^p_{A}, \pi^p_{R})$. Let $(\eta^p_{B, n}, \eta^p_{A, n}, \eta^p_{R, n})$ be the corresponding equilibrium points. Note that this family is tight (each equilibrium point is viewed as a product measure on a compact subset of $\R^3$). Therefore, by Prokhorov's Theorem, this family converges weakly along a subsequence to some $(\eta^p_{B}, \eta^p_{A}, \eta^p_{R})$. By Theorem 2 in \citet{Milgrom1985}, $(\eta^p_{B}, \eta^p_{A}, \eta^p_{R})$ is an equilibrium point. Since the utility functions are bounded and continuous, the corresponding tuple for this equilibrium is $(k^p, \pi^p_{A}, \pi^p_{R})$ and therefore $(k^p, \pi^p_{A}, \pi^p_{R}) \in \mathcal{E}^p$. Now define the function 
\[G(k, \pi_A, \pi_R; p) := (p + \delta \pi_A)(1 - F_1(k)) + \delta \pi_R F_1(k)\]
which is jointly continuous in $(k, \pi_A, \pi_R)$ and hence attains its maximum on the compact set $\mathcal{E}^p$. Let 
\[H(p) := \max_{(k, \pi_A, \pi_R) \in \mathcal{E}^p} G(k, \pi_A, \pi_R; p)\,.\]
We claim that $H(p)$ is upper semi-continuous. To see this, fix any $p_0$. Let $\{p_n\}_n$ converging to $p_0$ be a sequence along which we obtain $\displaystyle \limsup_{p \rightarrow p_0} H(p_0)$. Let \[(k^*_n, \pi^*_{A,n}, \pi^*_{R,n}) \in \argmax_{(k, \pi_A, \pi_R) \in \mathcal{E}^{p_n}}G(k, \pi_A, \pi_R; p_n)\,.\]
Let $(\eta^*_{B, n}, \eta^*_{A,n}, \eta^*_{R,n})$ be the corresponding equilibrium points. By the same argument as before, along a subsequence we know that $(\eta^*_{B, n}, \eta^*_{A,n}, \eta^*_{R,n})$ converges weakly to some $(\eta^*_{B,}, \eta^*_{A}, \eta^*_{R})$, which by \citet{Milgrom1985} is an equilibrium point of the game with parameter $p_0$. Then as before $(k^*_n, \pi^*_{A,n}, \pi^*_{R,n})$ also converges (along the subsequence) to $(k^*, \pi^*_{A}, \pi^*_{R})$ given by the equilibrium $(\eta^*_{B,}, \eta^*_{A}, \eta^*_{R})$.\footnote{One can view $(\eta^*_{B, n}, \eta^*_{A,n}, \eta^*_{R,n}, p_n)$ as a product measure on some compact subset of $\R^4$, and the claim follows by observing that the utility functions are bounded and jointly continuous in the actions and $p$.} By continuity of $G$, 
\[\limsup_{p \rightarrow p_0} H(p) = \lim_{n \rightarrow \infty} G(k^*_{n}, \pi^*_{A, n}, \pi^*_{R, n}; p_n) = G(k^*, \pi^*_A, \pi^*_R; p_0) \leq H(p_0) \]
where the last inequality is due to $(k^*, \pi^*_A, \pi^*_R)\in \mathcal{E}^{p_0}$, proving the claim.

Since $H$ is upper semi-continuous, and the set $[-(\overline{\theta}_1 + \overline{\theta}_2), \overline{\theta}_1 + \overline{\theta}_2]$ is compact, there exists a $p^*$ that maximizes $H(p)$. We now recover a PBE-$^\star$ of the original game. Let the seller offer $p^*$ in period $1$. For any $p$ offered by the seller, let $(k^p, \pi^p_{A}, \pi^p_{R})$ be a tuple in $\mathcal{E}^p$ that maximizes $G(\cdot; p)$. Let $(k^p, \eta^p_{A}, \eta^p_{R})$ be the corresponding equilibrium point. Let the buyer use strategy $x_1(\theta_1, p)= \1_{\theta_1 \geq k^p}$. Following histories $A, R$, the seller offers prices according to the mixed strategies $\eta^p_{A}, \eta^p_{R}$, respectively. Let her beliefs be $\mu_A = F_2( \ \cdot \ |\theta_1 \geq k^p)$ and $\mu_R = F_2(\ \cdot \ |\theta_1 \leq k^p)$. The buyer accepts in the second period if and only if his type is above the price.

By inspection, this is a PBE of the original game. Note that the construction guarantees that whenever $k = \overline{\theta}_1$, $\eta_A$ maximizes the seller's payoff with respect to $F_2(\, \cdot\,  | \theta_1 = \overline{\theta}_1, x_1 = 1)$ and whenever $k = \underline{\theta}_1$, $\eta_R$ maximizes the seller's payoff with respect to $F_2(\, \cdot\,  | \theta_1 = \underline{\theta}_1, x_1 = 0)$. Thus, the PBE is in fact a PBE-$^\star$.

\subsection{Proof of Part (i) of \Cref{thm:multi}} \label{app:existence-mp}

We show that there exists a PBE-$^\star$ of the subgame after the seller chooses not to commit in the first period, which immediately implies that the whole game also has a PBE-$^\star$ (see \Cref{app:existence} and \Cref{defn:PBEstar} for PBE-$^\star$). 
The proof proceeds in two steps. 

\paragraph{Step 1:} 
 We first show that the following game has a mixed-strategy PBE-$^\star$: for a fixed $T\geq2$, the seller posts a price $p$ in period $1$, and the buyer decides whether to accept; following histories $A, R$, the seller commits to prices $\{p^A_t\}_{t=2}^T, \{p^R_t\}_{t=2}^T$, respectively. The payoffs and the evolution of buyer's types are the same as in \Cref{sec:multi-mp}. 

\begin{lemma}\label{lemma:aux-mp}
There exists a PBE-$^\star$ of this game.
\end{lemma}

\begin{proof}

Consider the following auxiliary complete information one-shot game with $3$ players: $B$, $S_A$ , $S_R$. Player $B$ chooses $k \in [\underline{\theta}_1, \overline{\theta}_1]$ to maximize  
\[\E\Big[\big \{(\theta_1 - p) + \sum_{t=2}^T \delta^{t-1} (\theta_t - p^A_{t})_+ \big\} \1_{\theta_1 \geq k}\Big ] + \E\Big[\sum_{t=2}^T \delta^{t-1}(\theta_t - p^R_{t})_+ \1_{\theta_1 \leq k}\Big]\,.\]
Player $S_A$ chooses $\{p^A_t\}_{t=2}^T \in \displaystyle\prod_{t=2}^T[\underline{\theta}_t, \overline{\theta}_t]$ to maximize 
\[\sum_{t=2}^T\delta^{t-2}(1 - F_t(p^A_{t} | \theta_1 \geq k)) p^A_t \,.\]
Player $S_R$ chooses $\{p^R_t\}_{t=2}^T \in \displaystyle\prod_{t=2}^T[\underline{\theta}_t, \overline{\theta}_t]$ to maximize
\[\sum_{t=2}^T\delta^{t-2}(1 - F_t(p^R_t | \theta_1 \leq k)) p^R_t \,.\] The rest of proof is virtually the same as in \Cref{app:existence} and thus omitted. 
\end{proof}

\paragraph{Step 2:} We now consider the original subgame and show the following by induction.  
\begin{lemma}
The subgame after the seller chooses not to commit has a PBE-$^\star$. 
\end{lemma}

\begin{proof}
\textbf{Base case: } When $T = 1$, the claim holds trivially. 

\textbf{Inductive step: } Using the objects in \Cref{lemma:aux-mp} (see \Cref{app:existence} for the notation), we construct the strategies for the first two periods. Let the seller offer $p^*$ in period $1$.  At any history, after the seller offers some price $p$ in period $1$, the buyer responds according to the threshold rule with threshold $k^p$. Then following histories $A,R$, the seller offers subsequent prices according to the distributional strategies $\eta^p_A, \eta^p_R$ respectively; let her beliefs on $(\theta_1, \dots, \theta_T)$ be $F(\ \cdot \ |\theta_1 \geq k^p), F(\ \cdot \ |\theta_1 \leq k^p)$ respectively. The buyer accepts the offer in period $t$ whenever the price in period $t$ is weakly below $\theta_t$ for all $2 \leq t \leq T$. 

To complete the construction, consider the off-path history where the seller chooses not to commit in period $2$. Note that the seller's posterior belief about $(\theta_2, \dots, \theta_T)$ satisfies \Cref{asm:ar} after any history $(p, x_1)$ (since the buyer's first-period strategy is a threshold rule). Thus, by the inductive hypothesis, there exists a PBE-$^\star$ for the $(T-1)$-period game in which the buyer's type process is the posterior process $(\tilde{\theta}_2, \dots, \tilde{\theta}_T)$ (conditional on $\theta_1 \geq k^p$ or $\theta_1 \leq k^p$), after any history $(p, x_1)$ and the seller choosing not to commit. Use this PBE-$^\star$ to complete the construction of the strategies for the seller and buyer. 

Given the seller's strategy, all types of the buyer are playing optimally by construction. By part (ii) of \Cref{thm:multi} applying to the game starting from period $2$, we know that after any history $(p, x_1)$ the seller weakly prefers committing to some $\{p_2, \dots, p_T\}$. Then by the construction of the objects in \Cref{lemma:aux-mp}, the seller has no profitable deviation and has the required off-path beliefs. This concludes the inductive step. 
\end{proof}

\section{D1 Criterion}  \label{app:d1}

In this appendix, we provide a notion of D1 criterion in the spirit of \citet{Banks1987} for our game. As noted in \Cref{sec:multi-mp}, the game defined in \Cref{sec:main} is a subgame of the modified multiple-period game defined in \Cref{sec:multi-mp}. Thus, for notational convenience, we define all our notions in the multiple-period model (and allow the transition kernels of the type process to depend on the consumption in the previous period to accommodate \Cref{sec:comps}). We also maintain the assumptions in \Cref{sec:comps} for $T = 2$ (which nest those in \Cref{sec:main}), and the assumptions in \Cref{sec:multi-mp} for $T > 2$. Recall that we use PBE, PBE-D, PBE-$^\star$ to refer to the standard notion of PBE, PBE with D1 refinement (to be defined), and PBE with the particular class of off-path beliefs as in \Cref{defn:PBEstar}.

\subsection{Definition of D1 Criterion}

Fix any PBE. In any period $t$ and at any buyer's history $\hat{h}_t$, a price $p_t$ is posted by the seller. The buyer of type $\theta^t := (\theta_1, \dots, \theta_t)$ takes an action $x_t \in \{0, 1\}$. The seller observes this action, updates her belief, and moves to the continuation game. Let $U^*_{t}(\theta^{t}, \hat{h}_t)$ be type $\theta^t$'s continuation utility in this PBE (including the payoff from this period). We suppress the dependency on $\hat{h}_t$ whenever it is clear. For a given (potentially mixed) strategy $\alpha$ of the seller over the future plays in the continuation game, let $U_t(\theta^{t}, \alpha, \hat{h}_t)$ be the continuation utility of type $\theta^t$ best responding to $\alpha$. As explained in \Cref{lem:threshold-mp}, the dependency on $\theta^t$ reduces to the dependency on $\theta_t$. The seller also only makes inferences about $\theta_t$ from $x_t$. 

Now fix any $\hat{h}_t$ at which the PBE prescribes all types of the buyer to take the same action. Let $\mu_t$ be the prior belief on $\theta_{t}$ prescribed by the PBE for the seller before observing the buyer's action.  Let $x_t$ be the off-path action that the buyer could take. Let $BR(\Theta_t', \hat{h}_t, x_t)$ denote the set of possible future plays $\alpha$ that are best responses to some buyer's strategy, under some belief $\nu_t$ on $\theta_t$ with the restriction that $\nu_t$ is only supported on $\Theta_t' \subseteq \supp(\mu_t)$. Now let 
\[D(\theta_t, \Theta_t' , x_t) := \Big\{ \alpha \in BR(\Theta_t', x_t) \text{ s.t. } U_{t}^*(\theta_{t}) < U_{t}(\theta_t, \alpha , x_t)\Big\} \]
denote the set of seller's possible strategies for the future play under some belief concentrated on $\Theta_t'$ that makes type $\theta_t$ strictly prefer $x_t$ to the equilibrium play. Similarly, define 
\[D^0(\theta_t, \Theta_t' , x_t) := \Big\{ \alpha \in BR(\Theta_t', x_t) \text{ s.t. } U_{t}^*(\theta_{t}) = U_{t}(\theta_t, \alpha , x_t)\Big\} \,.\]
Let $\hat{\Theta}_t = \supp(\mu_t)$.  A type $\theta_t$ is deleted if there is another type $\theta'_t \in \hat{\Theta}_t$ such that 
\[\Big[ D(\theta_t, \hat{\Theta}_t,x_t) \cup D^0(\theta_t, \hat{\Theta}_t,x_t)\Big] \subset D(\theta'_t, \hat{\Theta}_t,x_t)\,.\]
Let $\Theta^*_t(x_t)$ be the set of surviving types until this process stops. We say that the PBE does not survive the D1 criterion if there exists some $\theta_t \in \Theta^*_t(x_t)$ such that 
\[U^*_t(\theta_t) < \inf_{\alpha \in \mathcal{A}(\Theta^*_t(x_t), x_t)} U_t(\theta_t, \alpha, x_t)\]
where $\mathcal{A}(\Theta^*_t(x_t), x_t)$ is the set of seller's strategies for the future play that belong to a PBE-D of the continuation game with some prior belief about $\theta_t$ supported only on $\Theta^*_t(x_t)$. Note that this definition is recursive. In the last period, after the buyer's choice, the seller no longer makes any decision, and so this set is $\varnothing$. In period $T-1$, this set is the set of seller's strategies that belong to a standard PBE of the one-period continuation game. 

In a signaling game, the above definition reduces to the usual definition, because there is no continuation game. In particular, the set of mixed best responses $BR(\cdot)$ does not depend on what the sender does subsequently; the set of receiver's continuation equilibrium responses $\mathcal{A}(\cdot)$  reduces to the set of mixed best responses. 

\subsection{PBE-D and PBE-$^\star$}

In this section, we show that the PBE-D in our setting are essentially PBE-$^\star$. This provides a foundation for the particular class of off-path beliefs we focus on and confirms the intuition that the D1 criterion effectively requires the seller to believe the buyer is of the type that would benefit the most upon an unexpected deviation.

\begin{lemma}\label{lem:survivetype}
At any history $\hat{h}_t$, the set of D1 surviving types after an off-path rejection and acceptance are given by $\Theta^*_t(0) = \{\underline{\theta}_t(\mu_t)\}, \Theta^*_t(1) = \{\overline{\theta}_t(\mu_t)\}$, where $ \underline{\theta}_t(\mu_t) = \min\{\emph{supp}(\mu_t)\}, \overline{\theta}_t(\mu_t) = \max\{\emph{supp}(\mu_t)\}$ and $\mu_t$ is the seller's belief on $\theta_t$ at history $h_t$.
\end{lemma}

\begin{proof}
We prove the case for an unexpected off-path rejection. The other case follows by a symmetric argument. Note that the mixed best response set $BR(\cdot)$ for the seller contains all feasible future strategies since in the definition we allow for the buyer to use arbitrary strategy for the continuation play.\footnote{If the buyer rejects all future offers, then any price path of the seller is in the $BR(\cdot)$ set.} For any seller's strategy $\alpha$ fixed, by \Cref{lem:threshold-mp} (and \Cref{lem:h3} for the case of $T = 2$), we have that
\[H(\theta_t, \alpha) :=  U^*_{t}(\theta_t) - U_{t}(\theta_t, \alpha, 0)\]
is a strictly increasing and continuous function of $\theta_t$. This implies that for any $\hat{\Theta}_t$ and any $\theta_t' < \theta_t \in \hat{\Theta}_t$, 
\[\Big[ D(\theta_t, \hat{\Theta}_t, 0) \cup D^0(\theta_t, \hat{\Theta}_t, 0)\Big] \subseteq D(\theta'_t, \hat{\Theta}_t, 0)\,.\]
Our hope is to delete the types from the top to the bottom. This requires a strict set inclusion instead of the weak one. Consider the particular class of strategy prescribing the sale of all future items at a constant price $\hat{p}$. Note that $H(\theta_t, \cdot)$ is a continuous function of $\hat{p}$ and for high enough $H(\theta_t, \hat{p}) = U^*_t(\theta_t) \geq 0$. Moreover, for any $\hat{p}$ low enough (potentially negative), $H(\theta_t, \hat{p}) < 0$. Therefore, there exists $\hat{p}^*(\theta_t)$ such that $H(\theta_t,\hat{p}^*(\theta_t)) = 0$. This strategy is in the set  $D^0(\theta_t, \hat{\Theta}_t, 0)$  for any $\theta_t$  and $\hat{\Theta}_t$. Fix any $\theta_t$ and any $\theta_t' < \theta_t$. Consider the seller's strategy of selling all future items at price $\hat{p}^*(\theta_t) + \epsilon$ for $\epsilon > 0$ small enough. This strategy is not in $ D(\theta_t, \hat{\Theta}_t, 0) \cup D^0(\theta_t, \hat{\Theta}_t, 0)$ by construction, but this strategy would be in $D(\theta'_t, \hat{\Theta}_t, 0)$ for $\epsilon$ small enough. This is because $H(\cdot, \cdot)$ is strictly increasing in the first argument and continuous in the second. This argument then shows that for any $\hat{\Theta}_t$ and any $\theta'_t < \theta_t \in \hat{\Theta}_t$, 
\[\Big[ D(\theta_t, \hat{\Theta}_t, 0) \cup D^0(\theta_t, \hat{\Theta}_t, 0)\Big] \subset D(\theta'_t, \hat{\Theta}_t, 0)\,.\]
Iterating this gives $\Theta^*_t(0) = \{\underline{\theta}_t(\mu_t)\}$. 
\end{proof}

\begin{lemma}\label{lem:startod1}
Every PBE-$^\star$ is a PBE-D. 
\end{lemma}
\begin{proof}
We prove this by induction on $T$. 

\textbf{Base case:} When $T = 1$, these two concepts are both equivalent to standard PBE.

\textbf{Inductive step:} Suppose for any $\tilde{T} \leq T - 1$ periods of our game, every PBE-$^\star$ is a PBE-D. Consider a $T$-period game and fix a PBE-$^\star$.  Fix any history $\hat{h}_t$ where all types are prescribed to accept. Consider the off-path deviation $x_t = 0$. Note that by \Cref{lem:survivetype}, we have $\Theta^*_t(0)= \underline{\theta}_t(\mu_t)$. By the definition of PBE, 
\[U^*_t(\underline{\theta}_t(\mu_t)) \geq U_t( \underline{\theta}_t(\mu_t) , \tilde{\sigma}^* , 0)\]
where $\tilde{\sigma}^*$ is the seller's equilibrium continuation play. By the definition of PBE-$^\star$, $\tilde{\sigma}^*$ is part of a PBE-$^\star$ for the continuation game under the belief concentrated on $\underline{\theta}_t(\mu_t)$. Since the continuation game has $\leq T - 1$ periods, by the inductive hypothesis,  $\tilde{\sigma}^* \in \mathcal{A}(\underline{\theta}_t(\mu_t), 0)$. Therefore, 
\[U^*_t(\underline{\theta}_t(\mu_t))  \geq  U_t( \underline{\theta}_t(\mu_t) , \tilde{\sigma}^* , 0) \geq\inf_{\alpha \in \mathcal{A}(\underline{\theta}_t(\mu_t), 0)} U_t(\underline{\theta}_t(\mu_t), \alpha, 0) \]
which shows that at any such history the PBE-$^\star$ passes the D1 test. The same argument holds for an unexpected acceptance. This proves the inductive step. 
\end{proof}

\begin{lemma}\label{lem:d1tostar}
For every PBE-D, there exists an outcome-equivalent PBE-$^\star$. 
\end{lemma}
\begin{proof}
We prove this by induction on $T$. 

\textbf{Base case:} When $T = 1$, these two concepts are both equivalent to standard PBE. 

\textbf{Inductive step:} Suppose for any $\tilde{T} \leq T - 1$ periods of our game, for every  PBE-D, there exists an outcome-equivalent PBE-$^\star$. 

Consider a $T$-period game and fix a PBE-D.  Fix any history $\hat{h}_t$ where all types are prescribed to accept. Consider the off-path deviation $x_t = 0$. Note that by \Cref{lem:survivetype}, we have $\Theta^*_t(0)= \underline{\theta}_t(\mu_t)$. By the definition of PBE-D, 
\[U^*_t(\underline{\theta}_t(\mu_t))  \geq\inf_{\alpha \in \mathcal{A}(\underline{\theta}_t(\mu_t), 0)} U_t(\underline{\theta}_t(\mu_t), \alpha, 0) \,.\]
Since the continuation game has $\leq T - 1$ periods, by the inductive hypothesis, 
\[\inf_{ \alpha \in \mathcal{A}(\underline{\theta}_t(\mu_t), 0)} U_t(\underline{\theta}_t(\mu_t), \alpha, 0)  = \inf_{\beta \in \mathcal{B}(\underline{\theta}_t(\mu_t), 0)} U_t(\underline{\theta}_t(\mu_t), \beta , 0)  \]
where $\mathcal{B}(\underline{\theta}_t(\mu_t), 0)$ is the collection of seller's strategies that belong to a PBE-$^\star$ of the continuation game that is outcome-equivalent to a PBE-D in $\mathcal{A}(\underline{\theta}_t(\mu_t), 0)$. 

If $T = 2$, then in the continuation game, PBE-D, PBE-$^*$, and PBE exist and coincide. They are specified by the distributions of prices that the seller posts. Let $p^*$ be the highest price the seller posts across the equilibria. Replace the continuation strategy profile given in the original PBE-D by letting the seller post $p^*$ and buyer accept if $\theta_2 \geq p^*$, with the seller's off-path belief assigned to be concentrated on $\underline{\theta}_t(\mu_t)$.

Now suppose $T > 2$. Note that the induction proof of part (ii) of \Cref{thm:multi} can be done for PBE-$^\star$ without referring to PBE-D. This implies that the PBE-$^\star$ corresponding to $\mathcal{B}(\underline{\theta}_t(\mu_t), 0)$ are all outcome-equivalent. The existence of such a PBE-$^\star$ follows from the proof in \Cref{app:existence-mp} (and thus $\mathcal{B}(\underline{\theta}_t(\mu_t), 0)$ and $\mathcal{A}(\underline{\theta}_t(\mu_t), 0)$ are non-empty). Now select any such PBE-$^\star$ of the continuation game. Replace the continuation strategy profile given in the original PBE-D by this PBE-$^\star$, with the seller's off-path belief assigned to be concentrated on $\underline{\theta}_t(\mu_t)$. 

In both cases ($T=2$ or $T>2$), we claim that the resulting strategy profile for the whole game is still a PBE. To see this, note that by construction in both cases the replacement leaves $\underline{\theta}_t(\mu_t)$ a payoff of $\inf_{ \alpha \in \mathcal{A}(\underline{\theta}_t(\mu_t), 0)} U_t(\underline{\theta}_t(\mu_t), \alpha, 0) \leq U^*_t(\underline{\theta}_t(\mu_t))$. So it remains optimal for type $\underline{\theta}_t(\mu_t)$ to play the prescribed strategy of acceptance. By the threshold-rule results (\Cref{lem:h3} and \Cref{lem:threshold-mp}), it is then optimal for all types above $\underline{\theta}_t(\mu_t)$ to continue accepting the offer. Now consider one layer above where the seller is choosing what price to offer in period $t$. Because this replacement only happens at an off-path history and the buyer uses the same threshold rule as before for this period, it is optimal for the seller to follow the originally prescribed strategy as well. It then follows that after this replacement the strategy profile is still a PBE of the whole game, and has the same equilibrium outcome as the original PBE-D. The same procedure also works for any history $\hat{h}_t$ at which all types of the buyer are prescribed to reject. 

Now for each period $t \in \{1, \dots, T\}$, for each $\hat{h}_t$, iteratively apply the above procedure whenever (i) all types of the buyer are prescribed to accept or reject at $\hat{h}_t$  and (ii) the assigned off-path belief is not already concentrated on $\underline{\theta}_t(\mu_t)$ or $\overline{\theta}_t(\mu_t)$.  It is evident that when the process stops, we have found a PBE-$^\star$ that is outcome-equivalent to the given PBE-D. This concludes the inductive step. 
\end{proof}

\bibliographystyle{ecta} 
\bibliography{library}

\begin{thebibliography}{42}
\newcommand{\enquote}[1]{``#1''}
\expandafter\ifx\csname natexlab\endcsname\relax\def\natexlab#1{#1}\fi

\bibitem[\protect\citeauthoryear{Acquisti, Taylor, and Wagman}{Acquisti
  et~al.}{2016}]{Acquisti2016}
\textsc{Acquisti, A., C.~Taylor, and L.~Wagman} (2016): \enquote{{The Economics
  of Privacy},} \emph{Journal of Economic Literature}, 54(2), 442--492.

\bibitem[\protect\citeauthoryear{Acquisti and Varian}{Acquisti and
  Varian}{2005}]{Acquisti2005}
\textsc{Acquisti, A. and H.~R. Varian} (2005): \enquote{{Conditioning Prices on
  Purchase History},} \emph{Marketing Science}, 24(3), 367--381.

\bibitem[\protect\citeauthoryear{Aviv and Pazgal}{Aviv and
  Pazgal}{2008}]{aviv2008optimal}
\textsc{Aviv, Y. and A.~Pazgal} (2008): \enquote{Optimal Pricing of Seasonal
  Products in the Presence of Forward-looking Consumers,} \emph{Manufacturing
  \& Service Operations Management}, 10(3), 339--359.

\bibitem[\protect\citeauthoryear{Banks and Sobel}{Banks and
  Sobel}{1987}]{Banks1987}
\textsc{Banks, J.~S. and J.~Sobel} (1987): \enquote{{Equilibrium Selection in
  Signaling Games},} \emph{Econometrica}, 55(3), 647--661.

\bibitem[\protect\citeauthoryear{Baron and Besanko}{Baron and
  Besanko}{1984}]{Baron1984}
\textsc{Baron, D.~P. and D.~Besanko} (1984): \enquote{{Regulation and
  Information in a Continuing Relationship},} \emph{Information Economics and
  Policy}, 1(3), 267--302.

\bibitem[\protect\citeauthoryear{Battaglini}{Battaglini}{2005}]{Battaglini2005}
\textsc{Battaglini, M.} (2005): \enquote{{Long-term Contracting with Markovian
  Consumers},} \emph{American Economic Review}, 95(3), 637--658.

\bibitem[\protect\citeauthoryear{Bergemann and V{\"{a}}lim{\"{a}}ki}{Bergemann
  and V{\"{a}}lim{\"{a}}ki}{2019}]{Bergemann2019}
\textsc{Bergemann, D. and J.~V{\"{a}}lim{\"{a}}ki} (2019): \enquote{{Dynamic
  Mechanism Design: An Introduction},} \emph{Journal of Economic Literature},
  57(2), 235--274.

\bibitem[\protect\citeauthoryear{Bonatti and Cisternas}{Bonatti and
  Cisternas}{2020}]{Bonatti2020}
\textsc{Bonatti, A. and G.~Cisternas} (2020): \enquote{{Consumer Scores and
  Price Discrimination},} \emph{Review of Economic Studies}, 87(2), 750--791.

\bibitem[\protect\citeauthoryear{Cachon and Feldman}{Cachon and
  Feldman}{2015}]{cachon2015price}
\textsc{Cachon, G.~P. and P.~Feldman} (2015): \enquote{Price Commitments with
  Strategic Consumers: Why It Can Be Optimal to Discount More Frequently \dots
  Than Optimal,} \emph{Manufacturing \& Service Operations Management}, 17(3),
  399--410.

\bibitem[\protect\citeauthoryear{Courty and Li}{Courty and
  Li}{2000}]{Courty2000}
\textsc{Courty, P. and H.~Li} (2000): \enquote{{Sequential Screening},}
  \emph{Review of Economic Studies}, 67(4), 697--717.

\bibitem[\protect\citeauthoryear{Crémer}{Crémer}{1984}]{Cremer1984}
\textsc{Crémer, J.} (1984): \enquote{On the Economics of Repeat Buying,}
  \emph{RAND Journal of Economics}, 15(3), 396--403.

\bibitem[\protect\citeauthoryear{Deb and Said}{Deb and
  Said}{2015}]{deb2015dynamic}
\textsc{Deb, R. and M.~Said} (2015): \enquote{Dynamic Screening with Limited
  Commitment,} \emph{Journal of Economic Theory}, 159, 891--928.

\bibitem[\protect\citeauthoryear{Devanur, Peres, and Sivan}{Devanur
  et~al.}{2019}]{devanur2019perfect}
\textsc{Devanur, N.~R., Y.~Peres, and B.~Sivan} (2019): \enquote{Perfect
  Bayesian Equilibria in Repeated Sales,} \emph{Games and Economic Behavior},
  118, 570--588.

\bibitem[\protect\citeauthoryear{Doval and Skreta}{Doval and
  Skreta}{2019}]{Doval2019}
\textsc{Doval, L. and V.~Skreta} (2019): \enquote{{Optimal Mechanism for the
  Sale of a Durable Good},} \emph{SSRN Electronic Journal}.

\bibitem[\protect\citeauthoryear{Doval and Skreta}{Doval and
  Skreta}{2020}]{Doval2020}
---\hspace{-.1pt}---\hspace{-.1pt}--- (2020): \enquote{{Mechanism Design with
  Limited Commitment},} \emph{SSRN Electronic Journal}.

\bibitem[\protect\citeauthoryear{Einav, Farronato, Levin, and Sundaresan}{Einav
  et~al.}{2018}]{einav2018auctions}
\textsc{Einav, L., C.~Farronato, J.~Levin, and N.~Sundaresan} (2018):
  \enquote{Auctions versus Posted Prices in Online Markets,} \emph{Journal of
  Political Economy}, 126(1), 178--215.

\bibitem[\protect\citeauthoryear{Elmaghraby and Keskinocak}{Elmaghraby and
  Keskinocak}{2003}]{elmaghraby2003dynamic}
\textsc{Elmaghraby, W. and P.~Keskinocak} (2003): \enquote{Dynamic Pricing in
  the Presence of Inventory Considerations: Research Overview, Current
  Practices, and Future Directions,} \emph{Management Science}, 49(10),
  1287--1309.

\bibitem[\protect\citeauthoryear{Eso and Szentes}{Eso and
  Szentes}{2007}]{Eso2007}
\textsc{Eso, P. and B.~Szentes} (2007): \enquote{{Optimal Information
  Disclosure in Auctions and the Handicap Auction},} \emph{Review of Economic
  Studies}, 74(3), 705--731.

\bibitem[\protect\citeauthoryear{Freixas, Guesnerie, and Tirole}{Freixas
  et~al.}{1985}]{freixas1985planning}
\textsc{Freixas, X., R.~Guesnerie, and J.~Tirole} (1985): \enquote{Planning
  under Incomplete Information and the Ratchet Effect,} \emph{Review of
  Economic Studies}, 52(2), 173--191.

\bibitem[\protect\citeauthoryear{Fudenberg and Villas-Boas}{Fudenberg and
  Villas-Boas}{2006}]{Fudenberg2006}
\textsc{Fudenberg, D. and J.~M. Villas-Boas} (2006): \enquote{{Behavior-Based
  Price Discrimination and Customer Recognition},} \emph{Handbook on Economics
  and Information Systems}, 377 -- 436.

\bibitem[\protect\citeauthoryear{Fudenberg and Villas-Boas}{Fudenberg and
  Villas-Boas}{2012}]{Fudenberg2012}
---\hspace{-.1pt}---\hspace{-.1pt}--- (2012): \enquote{{Price Discrimination in
  the Digital Economy},} \emph{Oxford Handbook of the Digital Economy},
  254--272.

\bibitem[\protect\citeauthoryear{Gerardi and Maestri}{Gerardi and
  Maestri}{2020}]{Gerardi2020}
\textsc{Gerardi, D. and L.~Maestri} (2020): \enquote{{Dynamic Contracting with
  Limited Commitment and the Ratchet Effect},} \emph{Theoretical Economics},
  15(2), 583--623.

\bibitem[\protect\citeauthoryear{Glicksberg}{Glicksberg}{1952}]{Glicksberg1952}
\textsc{Glicksberg, I.~L.} (1952): \enquote{{A Further Generalization of the
  Kakutani Fixed Theorem, with Application to Nash Equilibrium Points},}
  \emph{Proceedings of the American Mathematical Society}, 3(1), 170--170.

\bibitem[\protect\citeauthoryear{Gupta and Pathak}{Gupta and
  Pathak}{2014}]{gupta2014machine}
\textsc{Gupta, R. and C.~Pathak} (2014): \enquote{A Machine Learning Framework
  for Predicting Purchase by Online Customers Based on Dynamic Pricing,}
  \emph{Procedia Computer Science}, 36, 599--605.

\bibitem[\protect\citeauthoryear{Hart and Tirole}{Hart and
  Tirole}{1988}]{Hart1988}
\textsc{Hart, O.~D. and J.~Tirole} (1988): \enquote{{Contract Renegotiation and
  Coasian Dynamics},} \emph{Review of Economic Studies}, 55(4), 509.

\bibitem[\protect\citeauthoryear{Jing}{Jing}{2011}]{Jing2011}
\textsc{Jing, B.} (2011): \enquote{Pricing Experience Goods: The Effects of
  Customer Recognition and Commitment,} \emph{Journal of Economics \&
  Management Strategy}, 20(2), 451--473.

\bibitem[\protect\citeauthoryear{Kehoe, Larsen, and Pastorino}{Kehoe
  et~al.}{2020}]{kehoe2020dynamic}
\textsc{Kehoe, P.~J., B.~J. Larsen, and E.~Pastorino} (2020): \enquote{Dynamic
  Competition in the Era of Big Data,} \emph{Working Paper, Stanford University
  and Federal Reserve Bank of Minneapolis}.

\bibitem[\protect\citeauthoryear{Kennan}{Kennan}{2001}]{Kennan2001}
\textsc{Kennan, J.} (2001): \enquote{{Repeated Bargaining with Persistent
  Private Information},} \emph{Review of Economic Studies}, 68(4), 719--755.

\bibitem[\protect\citeauthoryear{Laffont and Tirole}{Laffont and
  Tirole}{1988}]{Laffont1988}
\textsc{Laffont, J.-J. and J.~Tirole} (1988): \enquote{{The Dynamics of
  Incentive Contracts},} \emph{Econometrica}, 56(5), 1153.

\bibitem[\protect\citeauthoryear{Liu, Mierendorff, Shi, and Zhong}{Liu
  et~al.}{2019}]{Liu2019}
\textsc{Liu, Q., K.~Mierendorff, X.~Shi, and W.~Zhong} (2019):
  \enquote{{Auctions with Limited Commitment},} \emph{American Economic
  Review}, 109(3), 876--910.

\bibitem[\protect\citeauthoryear{Milgrom and Weber}{Milgrom and
  Weber}{1982}]{Milgrom1982}
\textsc{Milgrom, P.~R. and R.~J. Weber} (1982): \enquote{{A Theory of Auctions
  and Competitive Bidding},} \emph{Econometrica}, 50(5), 1089.

\bibitem[\protect\citeauthoryear{Milgrom and Weber}{Milgrom and
  Weber}{1985}]{Milgrom1985}
---\hspace{-.1pt}---\hspace{-.1pt}--- (1985): \enquote{{Distributional
  Strategies for Games with Incomplete Information},} \emph{Mathematics of
  Operations Research}, 10(4), 619--632.

\bibitem[\protect\citeauthoryear{Pavan}{Pavan}{2017}]{Pavan2017}
\textsc{Pavan, A.} (2017): \enquote{{Dynamic Mechanism Design: Robustness and
  Endogenous Types},} in \emph{Advances in Economics and Econometrics}, 1--62.

\bibitem[\protect\citeauthoryear{Pavan, Segal, and Toikka}{Pavan
  et~al.}{2014}]{Pavan2014}
\textsc{Pavan, A., I.~Segal, and J.~Toikka} (2014): \enquote{{Dynamic mechanism
  design: A Myersonian approach},} \emph{Econometrica}, 82, 601--653.

\bibitem[\protect\citeauthoryear{Skreta}{Skreta}{2006}]{Skreta2006}
\textsc{Skreta, V.} (2006): \enquote{{Sequentially Optimal Mechanisms},}
  \emph{Review of Economic Studies}, 73(4), 1085--1111.

\bibitem[\protect\citeauthoryear{Stokey}{Stokey}{1979}]{Stokey1979}
\textsc{Stokey, N.~L.} (1979): \enquote{{Intertemporal Price Discrimination},}
  \emph{Quarterly Journal of Economics}, 93(3), 355--371.

\bibitem[\protect\citeauthoryear{Strulovici}{Strulovici}{2017}]{Strulovici2017}
\textsc{Strulovici, B.} (2017): \enquote{{Contract Negotiation and the Coase
  Conjecture: A Strategic Foundation for Renegotiation-Proof Contracts},}
  \emph{Econometrica}, 85(2), 585--616.

\bibitem[\protect\citeauthoryear{Su}{Su}{2007}]{su2007intertemporal}
\textsc{Su, X.} (2007): \enquote{Intertemporal Pricing with Strategic Customer
  Behavior,} \emph{Management Science}, 53(5), 726--741.

\bibitem[\protect\citeauthoryear{Taylor}{Taylor}{2004}]{taylor2004consumer}
\textsc{Taylor, C.~R.} (2004): \enquote{Consumer Privacy and the Market for
  Customer Information,} \emph{RAND Journal of Economics}, 35(4), 631--650.

\bibitem[\protect\citeauthoryear{Thompson and Chmura}{Thompson and
  Chmura}{2015}]{thompson2015loyalty}
\textsc{Thompson, F.~M. and T.~Chmura} (2015): \enquote{Loyalty Programs in
  Emerging and Developed Markets: The Impact of Cultural Values on Loyalty
  Program Choice,} \emph{Journal of International Marketing}, 23(3), 87--103.

\bibitem[\protect\citeauthoryear{Topkis}{Topkis}{1998}]{Topkis1998}
\textsc{Topkis, D.~M.} (1998): \emph{{Supermodularity and Complementarity}},
  Princeton University Press.

\bibitem[\protect\citeauthoryear{Villas-Boas}{Villas-Boas}{2004}]{Villas-Boas2004}
\textsc{Villas-Boas, J.~M.} (2004): \enquote{{Price Cycles in Markets with
  Customer Recognition},} \emph{RAND Journal of Economics}, 35(3), 486.

\end{thebibliography}

\end{appendices}

\end{document}